%% file: main.tex
\documentclass[
]{ceurart}

\usepackage{listings}
\usepackage[utf8]{inputenc}
\usepackage[T1]{fontenc} 
\usepackage{amsfonts}
\usepackage{amssymb}
\usepackage{amsthm}
\usepackage{amsmath}
\usepackage{latexsym}
\usepackage{stmaryrd}
\usepackage{wasysym}
\usepackage{mathrsfs}
\usepackage{graphicx}
\usepackage{rotating}
\usepackage{pifont}
\usepackage{booktabs}
\usepackage{array}
\usepackage{multirow}
\usepackage{tikz}
\usepackage{color}
\usepackage{xcolor}
\usepackage{enumitem}
\usepackage{xspace}
\usepackage{thm-restate} 
\usepackage{bussproofs}
\usepackage{tcolorbox}

\usepackage[edges]{forest}
\usepackage{varwidth}
\usepackage{adjustbox}

\newtheorem{theorem}{Theorem}
\newtheorem{definition}{Definition}

\newtheorem{corollary}[theorem]{Corollary}

\newtheorem{claim}{Claim}[theorem]
\input{macros_nnmdl}

\begin{document}

\copyrightyear{2022}
\copyrightclause{Copyright for this paper by its authors.
  Use permitted under Creative Commons License Attribution 4.0
  International (CC BY 4.0).}

\conference{ARQNL22: 4th International Workshop on Automated Reasoning in Quantified Non-Classical Logics, 11 August 2022, Haifa, Israel}

%\title{Reasoning in Non-normal Modal Description Logics (Preliminary Results)}
%\title{Reasoning in Non-normal Modal Description Logics with Varying Domains}
%\title{On Reasoning in Non-normal Modal \\ Description Logics}
\title{Reasoning in Non-normal Modal Description Logics}

%

%%
%% AUTHORS
\author[1]{Tiziano Dalmonte}[%
orcid=0000-0002-7153-0506,
email=tiziano.dalmonte@unibz.it,
%url=???,
]
%\address[1]{Free University of Bozen-Bolzano}

\author[1]{Andrea Mazzullo}[%
orcid=0000-0001-8512-1933,
email=andrea.mazzullo@unibz.it,
%url=???,
]
\address[1]{Free University of Bozen-Bolzano}

\author[2]{Ana Ozaki}[%
orcid=0000-0002-3889-6207,
email=ana.ozaki@uib.no,
url=https://www.uib.no/en/persons/Ana.Ozaki,
]
\address[2]{University of Bergen, Norway}

%%%%%%%%%%%%%%%%%%%%%%%%%%%%%%%%%%%%%%%%%%%%%%%%%%%%%%%%%%%%%%%%%%%%%%%%%%%%%%%%%%%%%%%%%%%%

\begin{abstract}
Non-normal modal logics, interpreted on neighbourhood models which generalise the usual relational semantics, have found application in several areas, such as epistemic, deontic, and coalitional reasoning. We present here preliminary results on reasoning in a family of modal description logics obtained by combining $\ALC$ with non-normal modal operators. First, we provide a framework of terminating, correct, and complete tableau algorithms to check satisfiability of formulas in such logics with the semantics based on varying domains. We then investigate the satisfiability problems in fragments of these languages obtained by restricting the application of modal operators to formulas only, and interpreted on models with constant domains, providing tight complexity results.
\end{abstract}


\maketitle
%
%\tableofcontents

%%%%%%%%%%%%%%%%%%%%%%%%%%%%%%%%%%%%%%%%%%%%%%%%%%%%%%%%%%%%%%%%%%%%%%%%%%%%%%%%%%%%%%%%%%%%

\input{sec_intro}
\input{sec_prel}

%%%%%%%%%%%%%%%%%%%%%%%%%%%%%%%%%%%%%%%%%%%%%%%%%%%%%%%%%%%%%%%%%%%%%%%%%%%%%%%%%%%%%%%%%%%%

\input{sec_reason}

\input{sec_frag}

%%%%%%%%%%%%%%%%%%%%%%%%%%%%%%%%%%%%%%%%%%%%%%%%%%%%%%%%%%%%%%%%%%%%%%

\input{sec_conc}

%%%%%%%%%%%%%%%%%%%%%%%%%%%%%%%%%%%%%%%%%%%%%%%%%%%%%%%%%%%%%%%%%%%%%%

%%% BIBLIOGRAPHY
%\newpage
\bibliography{bib_nnmdl}

\newpage
\appendix

\end{document}

%% file: macros_nnmdl.tex
\newcommand{\p}{\varphi}
\newcommand{\q}{\psi}

\newcommand{\role}{\ensuremath{r}\xspace}
\newcommand{\eset}{\emptyset}

\newcommand{\Int}{\ensuremath{\mathcal{I}}\xspace}

\newcommand{\dnot}{\ensuremath{\dot{\lnot}}\xspace}

\newcommand{\valuation}{\ensuremath{\nu}\xspace}

\newcommand{\propmodel}{\ensuremath{\mathcal{M}^{\sf P}}\xspace}
\newcommand{\propdomain}{\ensuremath{\mathcal{W}}\xspace}
\newcommand{\propneigh}{\ensuremath{\mathcal{N}}\xspace}
\newcommand{\propassign}{\ensuremath{\mathcal{V}}\xspace}

\newcommand{\con}{\ensuremath{\mathsf{con}}\xspace}
\newcommand{\conneg}{\ensuremath{\mathsf{con}_{\dot{\lnot}}}\xspace}
\newcommand{\for}{\ensuremath{\mathsf{for}}\xspace}
\newcommand{\forneg}{\ensuremath{\mathsf{for}_{\dot{\lnot}}}\xspace}
\newcommand{\rol}{\ensuremath{\mathsf{rol}}\xspace}

\newcommand{\fg}{\ensuremath{\mathsf{Fg}}\xspace}

\newcommand{\D}{\Diamond}

\newcommand{\B}{\Box}

\newcommand{\prop}[1]{\ensuremath{#1_{\sf prop} }\xspace} 
\newcommand{\elaxiom}{\ensuremath{\pi}\xspace} 
\newcommand{\consistent}[1]{\ensuremath{#1\text{-consistent}}\xspace}

\newcommand{\ALC}{\ensuremath{\smash{\mathcal{ALC}}}\xspace}

\newcommand{\NC}{\ensuremath{{\sf N_C}}\xspace}

\newcommand{\NR}{\ensuremath{{\sf N_R}}\xspace}

\newcommand{\NPr}{\ensuremath{{\sf N_P}}\xspace}

\newcommand{\NV}{\ensuremath{{\sf N_{V}}}\xspace}

\newcommand{\MLALC}[1]{\ensuremath{\smash{\mathit{ML}^{#1}_{\mathcal{ALC}}}}\xspace}

\newcommand{\MLnALC}{\ensuremath{\smash{\mathit{ML}^{n}_{\mathcal{ALC}}}}\xspace}

\newcommand{\MLnALCg}{\ensuremath{\smash{\mathit{ML}^{n\mid{\sf g}}_{\mathcal{ALC}}}}\xspace}
\renewcommand{\C}{\ensuremath{\smash{\mathbf{C}}}\xspace}

\newcommand{\LnALC}{\ensuremath{\mathbf{L}^{n}_{\ALC}}\xspace}

\newcommand{\LnALCg}{\ensuremath{\mathbf{L}^{n\mid{\sf g}}_{\ALC}}\xspace}

\newcommand{\EALCg}{\ensuremath{\smash{\mathbf{E}^{n\mid{\sf g}}_{\mathcal{ALC}}}}\xspace}
\newcommand{\CALCg}{\ensuremath{\smash{\mathbf{C}^{n\mid{\sf g}}_{\mathcal{ALC}}}}\xspace}
\newcommand{\NALCg}{\ensuremath{\smash{\mathbf{N}^{n\mid{\sf g}}_{\mathcal{ALC}}}}\xspace}

\newcommand{\MALCg}{\ensuremath{\smash{\mathbf{M}^{n\mid{\sf g}}_{\mathcal{ALC}}}}\xspace}

\newcommand{\NP}{\textsc{NP}}
\newcommand{\PSpace}{\textsc{PSpace}}
\newcommand{\ExpTime}{\textsc{ExpTime}}
\newcommand{\NExpTime}{\textsc{NExpTime}}

\newcommand{\Cmc}{\ensuremath{\mathcal{C}}\xspace}

\newcommand{\Fmc}{\ensuremath{\mathcal{F}}\xspace}

\newcommand{\Imc}{\ensuremath{\mathcal{I}}\xspace}

\newcommand{\Mmc}{\ensuremath{\mathcal{M}}\xspace}
\newcommand{\Nmc}{\ensuremath{\mathcal{N}}\xspace}

\newcommand{\Vmc}{\ensuremath{\mathcal{V}}\xspace}
\newcommand{\Wmc}{\ensuremath{\mathcal{W}}\xspace}

\newcommand{\N}{\ensuremath{\mathcal{N}_{i}}\xspace}
\newcommand{\W}{\ensuremath{\mathcal{W}}\xspace}
\newcommand{\M}{\ensuremath{\mathcal{M}}\xspace}

%% file: sec_intro.tex
\section{Introduction}

Contexts involving
%a certain degree of
epistemic and doxastic~\cite{Ago,Bal,Var1}, agency-based~\cite{Brown,Elg} and coalitional~\cite{Pau,Tro}, as well as deontic~\cite{AngEtAl,Gob,Wright}, reasoning capabilities populate the wide spectrum of settings where modal logics have found natural applications.
In such scenarios, modal operators can be used to represent and reason about what agents, or groups of agents, respectively know, believe, have the capability, or have the permission, to bring about.

The semantics of modal operators is usually given in terms of \emph{relational models}, based on frames consisting of a set of possible worlds equipped with suitable accessibility relations.
However, all the modal systems interpreted by means of this kind of semantics, known as \emph{normal}, validate principles that have been considered problematic or debatable for the aforementioned applications, leading to counterintuitive or unacceptable conclusions.
Among the unpleasant features discussed in the literature, one encounters for instance the problem of logical omniscience~\cite{Var1}, as well as a number of so-called paradoxes in the representation of agents' abilities~\cite{Elg} and obligations~\cite{Ross,Aqv,For}.

To avoid the unwanted consequences of the relational semantics, several \emph{non-normal} modal logics have been proposed and studied, tracing back to the seminal works by C.I. Lewis~\cite{CIL}, Lemmon~\cite{Lem}, Kripke~\cite{Kripke}, Scott~\cite{Sco}, Montague~\cite{Mon}, Segerberg~\cite{Seg}, and Chellas~\cite{Che}.
The semantics of such systems can be given in terms of \emph{neighbourhood models}, generalisations of the relational ones that were first introduced by Scott \cite{Sco} and Montague \cite{Mon}.
In this setting, a frame consists of a set of worlds, each of which is associated with a set of subsets of worlds. Since a subset of worlds can be thought as a proposition (that is true in those worlds), this means that every world in a neighbourhood model is assigned to a set of propositions, those considered necessary with respect to that world.
This semantics both generalises the relational one, and avoids the drawbacks of the latter, since the modal principles validated on relational frames that are deemed as problematic for epistemic, coalitional or deontic applications do not hold in general on neighbourhood models.

Non-normal modalities have been widely investigated as a way to extend propositional logic.
A further line of research focuses on the behaviour of modal operators interpreted on neighbourhood frames in combination with first-order logic.
In this direction, a few works have provided completeness results for first-order non-normal modal logics~\cite{Cos,CosPac}.
In addition, non-normal modal extensions of \emph{description logics}, seen as fragments of first-order logic with a good trade-off between expressive power and computational complexity, have been considered for knowledge representation applications~\cite{SeyErd09, DL19}, also in multi-agent coalitional settings~\cite{SeyJam09,SeyJam10}.

In this paper, we investigate
satisfiability of non-normal modal extensions of description logics.
In particular,
%we study the logics associated with the class of
%all neighbourhood frames,
%supplemented neighbourhood frames, 
%neighbourhood frames closed under intersection,
%and neighbourhood frames containing the unit
%(characterising, at the propositional level, the non-normal modal systems $\mathbf{E}$, $\mathbf{M}$, $\mathbf{C}$, and $\mathbf{N}$, respectively),
we study the logics
characterised by
%associated with
the class of all neighbourhood frames ($\mathbf{E}$),
supplemented neighbourhood frames ($\mathbf{M}$), 
neighbourhood frames closed under intersection ($\mathbf{C}$),
and neighbourhood frames containing the unit ($\mathbf{N}$),
and combine them with the prototypical $\ALC$ description logic. 
We provide a framework of terminating, correct, and complete tableau algorithms to check 
satisfiability in such logics interpreted
in neighbourhood models
with \emph{varying domains}
(in this kind of semantics, the domains of the interpretations
at each world
%representing worlds
can differ;
cf.~Section~\ref{sec:prelim} for details).
%regarding the semantics). 
%Our proofs for the tableau algorithm include a characterization 
%of neighbourhood models in terms of \emph{quasimodels}~\cite{?}, %cite yellow book
%a classical technique used in modal logic to prove complexity results.  
We then investigate the satisfiability problems in fragments of these languages obtained by restricting the application of modal operators to formulas only, and provide complexity upper bounds with \emph{constant domains}
(in this case 
%this is the more classical case in which
the domains of the interpretations
at every world
%representing the worlds
are the same).
We leave satisfiability checking procedures for non-restricted languages interpreted on models with constant domain as open problems.

%% file: sec_prel.tex
\section{Preliminaries}
\label{sec:prelim}

In this section, we provide preliminary definitions for non-normal modal description logics, first introducing their syntax, and then giving their semantics based on neighbourhood models.

%\subsection{Syntax}

\paragraph{Syntax} Let \NC and \NR be countably infinite and pairwise disjoint 
sets of \emph{concept names} and \emph{role names}
respectively.
An $\MLALC{n}$ \emph{concept} is an expression of the form %\nb{O: I was thinking about adding EL, but due to time constraints I will first focus on ALC with global GCIs. }
\[
C ::= A \mid \lnot C \mid C \sqcap C \mid \exists \role.C \mid \B_{i} C,
\]
%\[
%C, D ::= A \mid \lnot C \mid C \sqcap D \mid \exists \role.C \mid \B_{i} C,
%\]
where $A \in \NC$, $\role \in \NR$, and $\B_{i}$, with
$i \in I = \{ 1, \ldots, n \}$,
%$1 \leq i \leq n$,
are modal operators called \emph{boxes}.  
%For instance, the expression $\B (B \sqcap \exists R. \lnot A)$ is a \TALC~concept.
%
%An \emph{\MLALC{n} atom}
%\nb{M: Ho eliminato $\alpha$ qui}
%is
A \emph{concept inclusion} (\emph{CI}) is an expression of
the form $C \sqsubseteq D$,
where $C, D$ are $\MLnALC{}$ concepts.
%
%A \emph{\MLALC{} formula} $\p$ is defined inductively as in $\ML$ from the set of $\MLALC{}$ atoms.
An \emph{\MLALC{n} formula}
takes the form
%is an expression of the form
\[
\p ::= C \sqsubseteq D \mid \neg \p \mid \p \land \p \mid \B_{i} \p,
\]
%\[
%\p, \psi ::= C \sqsubseteq D \mid \neg \p \mid \p \land \psi \mid \B_{i} \p,
%\]
where
%$\pi$ is an \MLALC{n} atom,
%and
$i \in I$.
%$ I = \{ 1, \ldots, n \}$.
%$1 \leq i \leq n$.
%
We will use the following standard definitions for concepts:
$\bot := A \sqcap \lnot A$,
$\top :=  \lnot \bot$;
$\forall \role.C :=  \lnot \exists \role.\lnot C$;
$(C \sqcup D) :=  \lnot(\lnot C \sqcap \lnot D)$;
%$(C \Rightarrow D) := (\lnot C \sqcup D)$;
%$(C \Leftrightarrow D) := (C \Rightarrow D) \sqcap (D \Rightarrow C)$;
$\D_{i} C := \lnot \B_{i} \lnot C$ (operators $\D_{i}$ are called \emph{diamonds}).
%\nb{M: $\D$ primitivo?}
Concepts of the form $\B_{i} C$, $\D_{i} C$ are called \emph{modalised concepts}.
Analogous conventions also hold for formulas,
for which we set $\mathsf{true} := (\bot \sqsubseteq \top)$.
%In particular, we write $C \equiv D$ for $C \sqs D \land D \sqs C$.
%
%The language $\MLALC{1}$ is denoted simply by $\MLALC{}$.
%\nb{O: is this being used? \\ M: In the examples, but I changed to $\MLALC{1}$ there. Thanks}

%%%%%%%%%%%%%%%%%%%%%%%%%%%%%%%
%\subsection{Semantics}

%%%%%%%%%%%%%%%%%%%%%%%%%%%%%
%\subsubsection{Neighbourhood Semantics}
%\subsection{Neighbourhood Semantics for Propositional Modal Logic}
%\nb{M: Changed! neighbourhood frame $\to$ multimodal neighbourhood frame}
\paragraph{Semantics}
 A \emph{neighbourhood frame}, or simply \emph{frame},
%shortened to \emph{\Rframe} when $n$ is clear from the context)
is a pair
%\nb{M: Aggiunto $\Fmc$}
$\Fmc = ( \Wmc, \{\Nmc_i \}_{i \in I})$,
where 
$\Wmc$ is a non-empty set
of \emph{worlds}
and,
%$\Nmc$ is a function associating to
for each
%$1 \leq i \leq n$, $\Nmc_{i}$
$i \in I = \{1, \ldots, n\}$,
$\Nmc_{i} \colon \W \rightarrow 2^{2^{\Wmc}}$ is called a \emph{neighbourhood function}.
%$\N$ is a function $\W \longrightarrow 2^{2^{\Wmc}}$, called \emph{neighbourhood function}.
%
A frame is:
\emph{supplemented} if,
for all
$i \in I$,
$w\in \Wmc$, 
$\alpha,\beta\subseteq \Wmc$, $\alpha\in \N(w)$ and $\alpha\subseteq\beta$ implies $\beta\in \N(w)$;
%it is
\emph{closed under intersection} if,
for all
$i \in I$,
$w\in \Wmc$, $\alpha,\beta\subseteq \Wmc$, 
$\alpha\in \N(w)$ and $\beta\in \N(w)$ implies $\alpha\cap\beta\in \N(w)$;
and
%it
\emph{contains the unit} if,
for all
$i \in I$,
$w\in \Wmc, \Wmc \in \N(w)$.
%
%\label{sec:semantics}
An \emph{\MLALC{n} varying domain neighbourhood model}, or simply \emph{model}, based on a neighbourhood frame $\Fmc$ is a pair
$\Mmc = (\Fmc, \Int)$,
where
$\Fmc = (\Wmc,  \{\Nmc_i \}_{i \in I})$ is a neighbourhood frame
%neighbourhood frame,
%$\Delta$ is a non-empty set called \emph{domain of $\Mmc$},
and $\Imc$ is a function associating with every $w \in \Wmc$ an \emph{$\ALC$ interpretation}
%(or \emph{model})
$\Imc_{w} = (\Delta_{w}, \cdot^{\Imc_{w}})$,
with non-empty \emph{domain} $\Delta_{w}$,
%so that $\Delta = \bigcup_{w \in \Wmc} \Delta_{w}$,
and where $\cdot^{\Imc_{w}}$ is a function such that:
for all $A \in \NC$, $A^{\Imc_{w}} \subseteq \Delta_{w}$;
for all $\role \in \NR$, $\role^{\Imc_{w}} \subseteq \Delta_{w} {\times} \Delta_{w}$.
An \MLALC{n} \emph{constant domain neighbourhood model}
is defined in the same way, except that, for all $w,w'\in\Wmc$,
we have that $\Delta_{w}=\Delta_{w'}$.
%\nb{O: changed; M: thanks}
%
%A model $\M = \langle \Fmc, \Delta_{w}, \Imc \rangle$ is \emph{supplemented, closed under intersection, or contains the unit}
%\nb{M: Used? Check}
%if $\Fmc$ is supplemented, closed under intersection, or contains the unit, respectively.
%
%The latter condition means that the individual names are treated as \emph{rigid designators}.
%Moreover, since the domain $\Delta_{w}$ is shared by all worlds $w \in \Wmc$, we say that we accept the \emph{constant domain assumption}.
%
Given a model $\Mmc = (\Fmc, \Int)$ and a world $w \in \Wmc$ of $\Fmc$ (or simply \emph{$w$ in $\Fmc$}), the \emph{interpretation $C^{\Imc_{w}}$ of a concept $C$ in $w$}
%written $C^{\Imc_{w}}$,
is defined
%by taking:
as:
% for the 
%similarly to the 
%relational semantics, % case, 
%except from: % modalised concepts:
	\begin{align*}
		(\neg D)^{\Imc_{w}} &= \Delta_{w} \setminus D^{\Imc_{w}}, \\ 
		(D \sqcap E)^{\Imc_{w}} &= D^{\Imc_{w}} \cap E^{\Imc_{w}}, \\
		(\exists r.D)^{\Imc_{w}} &= \{d \in \Delta_{w} \mid \exists  e \in D^{\Imc_{w}}: (d,e) \in r^{\Imc_{w}}\},\\
		(\B_{i} D)^{\Imc_{w}} &= \{ d \in \Delta_{w} \mid \llbracket D \rrbracket^{\Mmc}_{d} \in \N(w) \},
	\end{align*}
%	\begin{gather*}
%		(\neg D)^{\Imc_{w}} = \Delta_{w} \setminus D^{\Imc_{w}}, \quad 
%		(D \sqcap E)^{\Imc_{w}} = D^{\Imc_{w}} \cap E^{\Imc_{w}}, \\
%		(\exists R.D)^{\Imc_{w}} = \{d \in \Delta_{w} \mid \exists  e \in D^{\Imc_{w}}: (d,e) \in R^{\Imc_{w}}\},\\
%		(\B_{i} D)^{\Imc_{w}} = \{ d \in \Delta_{w} \mid \llbracket D \rrbracket^{\Mmc}_{d} \in \N(w) \},
%	\end{gather*}
%	\begin{gather*}
%		(\neg C)^{\Imc_{w}} = \Delta_{w} \setminus C^{\Imc_{w}}, \quad 
%		(C \sqcap D)^{\Imc_{w}} = C^{\Imc_{w}} \cap D^{\Imc_{w}}, \\
%		(\exists R.C)^{\Imc_{w}} = \{d \in \Delta_{w} \mid \exists  e \in C^{\Imc_{w}}: (d,e) \in R^{\Imc_{w}}\},\\
%		(\B_{i} C)^{\Imc_{w}} = \{ d \in \Delta_{w} \mid \llbracket C \rrbracket^{\Mmc}_{d} \in \N(w) \},
%	\end{gather*}
%\nb{M: Check $\llbracket C \rrbracket^{\Mmc}_{d}$}
where, for all %\nb{M: changed, to check (it was $d \in \Delta_{w}$ before)} 
$d \in \bigcup_{w \in \Wmc} \Delta_{w}$, the set
$\llbracket D \rrbracket^{\Mmc}_{d} = \{ v \in \Wmc \mid  d \in D^{\Imc_{v}} \}$
is called the \emph{truth set of $D$ with respect to $d$}.
%$\llbracket C \rrbracket^{\Mmc}_{d} = \{ v \in \Wmc \mid  d \in C^{\Imc_{v}} \}$
%is called the \emph{truth set of $C$ with respect to $d$}.
% (note that for constant domains ).
We say that a concept $C$ is \emph{satisfied in $\Mmc$} if there is $w$ in $\Fmc$ such that $C^{\Imc_{w}} \neq \eset$, and that $C$ is \emph{satisfiable} (over varying or constant neighbourhood models, respectively) if there is a (varying or constant domain, respectively) neighbourhood model in which it is satisfied.
The \emph{satisfaction of an $\MLALC{n}$ formula~$\p$ in $w$ of $\Mmc$}, written $\Mmc, w  \models \p$, is defined
%analogously to relational semantics, and
as follows:
%\begin{gather*}
%	\Mmc, w \models C\sqsubseteq D \quad \text{iff} \quad C^{\Imc_{w}} \subseteq D^{\Imc_{w}},\\
%%	\qquad
%	\Mmc, w  \models \neg \psi \quad \text{iff} \quad \Mmc, w  \not \models \psi, \\
%	\Mmc, w  \models \psi \land \chi \quad \text{iff} \quad \Mmc, w  \models \psi \text{ and } \Mmc, w  \models \chi,\\
%%	\qquad
%	\Mmc, w  \models \B_{i} \psi \quad \text{iff} \quad \llbracket \psi \rrbracket^{\Mmc} \in \N(w),
% \end{gather*} 
\begin{align*}
	\Mmc, w \models C\sqsubseteq D &\text{\quad iff \quad} C^{\Imc_{w}} \subseteq D^{\Imc_{w}}, \\
	\Mmc, w  \models \neg \psi &\text{\quad iff \quad} \Mmc, w  \not \models \psi, \\
	\Mmc, w  \models \psi \land \chi &\text{\quad iff \quad} \Mmc, w  \models \psi \text{ and } \Mmc, w  \models \chi, \\
	\Mmc, w  \models \B_{i} \psi &\text{\quad iff \quad} \llbracket \psi \rrbracket^{\Mmc} \in \N(w),
 \end{align*} 
%\\
%%
%\begin{tabular}{ccc}
%	$\Mmc, w \models C\sqsubseteq D$ & \text{iff} & $C^{\Imc_{w}} \subseteq D^{\Imc_{w}},$ \\
%	$\Mmc, w  \models \neg \psi$ & \text{iff} & $\Mmc, w  \not \models \psi,$ \\
%	$\Mmc, w  \models \psi \land \chi$ & \text{iff} & $\Mmc, w  \models \psi \text{ and } \Mmc, w  \models \chi,$ \\
%	$\Mmc, w  \models \B_{i} \psi$ & \text{iff} & $\llbracket \psi \rrbracket^{\Mmc} \in \N(w),$
% \end{tabular} 
% \\
%%\begin{align*}
%%%	{\color{red}{\Mmf, w  \models A(a) \text{ \, iff \, } a^{I} \in A^{\Imc_{w}}, \qquad
%%%	\Mmf, w \models \role(a,b) \text{ \, iff \, } (a^{I},b^{I}) \in \role^{\Imc_{w}},}} \\
%%	\Mmc, w  \models \neg \p \text{ \, iff \, } \Mmc, w  \not \models \p, \qquad
%%	\Mmc, w  \models \p \land \psi \text{ \, iff \, } \Mmc, w  \models \p \text{ and } \Mmc, w  \models \psi, \\
%%%	\Mmf, w  \models \B_{i} \p \text{ \, iff \, } \text{for all $v \in \reldomain$: if $w \relations_{i} v$, then } \Mmf, v  \models \p.
%%	\Mmc, w \models C\sqsubseteq D  \text{ \, iff \, } C^{\Imc_{w}} \subseteq D^{\Imc_{w}}, \qquad
%%	\Mmc, w  \models \B_{i} \p  \text{ \quad iff \quad } \llbracket \p \rrbracket^{\Mmc} \in \N(w), 
%% \end{align*} 
%
where
$\llbracket \psi \rrbracket^{\Mmc} = \{ v \in \Wmc \mid \Mmc, v \models \psi \}$ is the \emph{truth set of $\psi$}.
%i.e. the set the worlds $v$ that satisfy $\psi$.
%%$\llbracket \psi \rrbracket^{\Mmc}$ denotes the set $\{ v \in \Wmc \mid \Mmc, v \models \psi \}$ of the worlds $v$ that satisfy $\psi$, called the \emph{truth set of $\psi$}.
As a consequence of the above definition, we obtain the following
condition for diamond formulas:
$\M, w \models \Diamond_{i} \psi$  iff  $\llbracket \neg \psi \rrbracket^{\M} \notin \N(w)$.
Given a
neighbourhood
frame $\Fmc = (\Wmc,  \{\Nmc_i \}_{i \in I})$
and a
neighbourhood
model $\Mmc = (\Fmc, \Imc)$,
we say that $\varphi$ is \emph{satisfied in $\Mmc$} if there is $w \in \Wmc$ such that
$\Mmc, w \models \varphi$,
and that $\p$ is \emph{satisfiable} (over varying or constant domain neighbourhood models, respectively) if it is satisfied in some (varying or constant domain, respectively) neighbourhood model.
Given a class of frames $\Cmc$, by the \emph{$\MLALC{n}$ formula satisfiability problem  on} (\emph{varying} or \emph{constant domain}, respectively) \emph{neighbourhood models based on a frame in $\Cmc$} we mean the problem of deciding whether an $\MLALC{n}$ formula is satisfied in a (varying or constant domain, respectively) neighbourhood model based on a frame in $\Cmc$.
In the following, let $\mathsf{Log} = \{ \mathbf{E}, \mathbf{M}, \mathbf{C}, \mathbf{N} 
%\mathbf{MC}, \mathbf{MN}, \mathbf{CN}
\}$. 
Given $\mathbf{L} \in \mathsf{Log}$, the \emph{$\LnALC$ formula satisfiability problem on} (\emph{varying} or \emph{constant domain}, respectively) \emph{neighbourhood models} is the $\MLALC{n}$ formula satisfiability problem on (varying or constant domain, respectively) neighbourhood models based on a frame in the class of:
\begin{itemize}
	\item  all neighbourhood frames, for $\mathbf{L} = \mathbf{E}$;
	\item  supplemented neighbourhood frames, for $\mathbf{L} = \mathbf{M}$;
	\item  neighbourhood frames closed under intersection, for $\mathbf{L} = \mathbf{C}$; and
	\item neighbourhood frames containing the unit, for $\mathbf{L} = \mathbf{N}$.
\end{itemize}

%%% EXAMPLE (not needed)
%For example, the formula satisfiability problem for ${\mathbf{C}}^{n}_{\ALC}$
%is the formula satisfiability problem in the class of neighbourhood frames closed under intersection. For every neighbourhood frame \Fmc in this class, it holds that 
%$\Fmc \models \Box_{i} C \sqcap \Box_{i} D \sqsubseteq \Box_{i} (C \sqcap D)$~\cite{DL19} .
%%while, for $\mathbf{L} \in \{ \mathbf{MC}, \mathbf{MN}, \mathbf{CN} \}$, the above properties are combined in the obvious way. 

%The \emph{formula satisfiability problem for} $\EnALC{n}$, $\MnALC{n}$,
%{\color{blue} ${\mathbf{C}}^{n}_{\ALC}$, ${\mathbf{N}}^{n}_{\ALC}$,}
%\nb{M: todo remove red}
%{\color{red}{and $\KnALC{n}$}} is the $\MLALC{n}$ formula satisfiability problem in the class of neighbourhood frames, supplemented neighbourhood frames, {\color{blue}{neighbourhood frames closed under intersection, neighbourhood frames containing the unit,}} {\color{red}{and R-frames}},
%respectively.
%

%% file: sec_reason.tex
\section{Tableaux for Non-normal Modal Description Logics}
%\section{Reasoning in Non-normal Modal Description Logics}

In this section, we provide terminating, sound and complete tableau algorithms to check satisfiability of formulas in varying domain neighbourhood models. The notation partly adheres to that of~\citeauthor{GabEtAl03}~\cite{GabEtAl03}, while the model construction in the soundness proof is based on the strategy of~\citeauthor{DalHyp}~\cite{DalHyp}.
%\textcolor{red}{To cite: \cite{DL19} our previous work on non-normal DL, \cite{DalHyp} for the definition of the model in Theorem 2}

%\input{subsec_ruletableau}
%\subsection{Tableaux for Non-normal Modal Description Logics}

%Adhering to the notation used in~\cite{SeyErd09},
We require the following preliminary notions.
%Given an $\MLALC{n}$ formula $\p$, we denote by $\conneg(\p)$ and $\forneg(\p)$ the set of subconcepts and subformulas of $\p$, respectively.
For a concept or formula $\gamma$, we denote by $\dot{\lnot}\gamma$ the negation of $\gamma$ put in \emph{negation normal form} (\emph{NNF}), defined as usual.
Given an $\MLALC{n}$ formula $\p$, we assume without loss of generality that $\p$ is in NNF, it contains CIs only of the form $\top \sqsubseteq C$, and every concept occurring in $\p$ is also in NNF.
%\nb{M: todo def. weight}
We define the \emph{weight} $|C|$ of a concept $C$ in NNF as follows: $|A| = |\lnot A| = 0$; $|\exists r.D| = |\forall r.D| = |\Diamond_{i}D| = |\Box_{i}D| = |D| + 1$; $|D \sqcap E| = |D \sqcup E| = |D| + |E| + 1$. The \emph{weight $|\p|$} of a formula $\p$ in NNF is defined as: $| (C \sqsubseteq D) | = | \lnot (C \sqsubseteq D) | = 0$; $\Box_{i} \psi = | \psi | + 1$; $| \psi \land \chi | = | \psi \lor \chi | = | \psi | + | \chi | + 1$.
Observe that, for a concept or formula $\gamma$, we have that $| \gamma | = | \dnot \gamma |$.
We denote by $\con(\p)$ and $\for(\p)$ the set of subconcepts  and subformulas of $\p$, respectively, and then we set
$\conneg(\p) = \con(\p) \cup \{ \dot{\lnot}C \mid C \in \con(\p) \}$ and $\forneg(\p) =  \for(\p) \cup \{ \dot{\lnot}\psi \mid \psi \in \for(\p) \}$.
The set $\mathsf{rol}(\p)$ is the set of role names occurring in $\p$.
%notion of a tableau for an $\MLALC{n}$ formula.
Let 
%$\fg(\p) = \forneg(\p) \cup \conneg(\p) \cup \rol(\p) \cup \{ \top \}$.
$\fg(\p) = \forneg(\p) \cup \conneg(\p) \cup \rol(\p)$.
Note that, %as a consequence of 
by our assumption on the form of CIs in $\p$, we have $\top\in\conneg(\p)$.

%\nb{M: to check - merged T. defs. here (+ small changes)}
%{\color{blue}{
		Moreover, let $\NV$ be a countable set of \emph{variables}, well-ordered by the relation $<$, and let $\mathsf{N_{L}}$ be a countable set of \emph{labels}.
		%well-ordered by the relation $\ll$.
		Given an $\MLnALC$ formula $\p$, an \emph{$n$-labelled constraint for $\p$} takes the form $n: \psi$, or $n: C(x)$, or $n: r(x, y)$, where $n \in \mathsf{N_{L}}$, $\psi \in \forneg(\p)$, $x \in \NV$, 
		%$C \in \conneg(\p) \cup \{ \top \}$, 
		$C \in \conneg(\p)$, 
		and $r \in \rol(\p)$.
		An \emph{$n$-labelled constraint system for $\p$} is a set $S_{n}$ of $n$-labelled constraints for $\p$.
		(A \emph{labelled constraint for $\p$} is an $n$-labelled constraint for $\p$, for some $n \in \mathsf{N_{L}}$, and similarly for a \emph{labelled constraint system for $\p$}).
		%A \emph{completion set} $\mathbf{T}$ is a
		%%set of labelled constraints for $\p$.
		%(non-empty) union of labelled constraint systems for $\p$.
		A
		%\nb{M: changed - simplified (add `non-empty'?)}
		\emph{completion set} $\mathbf{T}$ is a non-empty
		union
		%set
		of labelled constraint system for $\p$,
		and we set $\mathsf{L}_{\mathbf{T}} = \{ n \in \mathsf{N_{L}} \mid S_{n} \in \mathbf{T} \}$.
		%Finally, a set $\mathbf{T} = \bigcup^{m}_{n = 0} S_{n}$,  for some $m \in \mathsf{N_{L}}$, is called \emph{completion set}, where each $S_{n}$ is an $n$-labelled constraint system for $\p$.
		%When no confusion can arise, we omit `$n$' and speak of `labelled constraint for $\p$' or of `labelled constraint system for $\p$'.
		
		Concerning variables, we adopt the following terminology.
		%\nb{M: todo add blocking (cf. YB)}
		A variable $x$ \emph{occurs in $S_{n}$} if $S_{n}$ contains $n$-labelled constraints of the form $n: C(x)$ or $n: r(\tau,\tau')$, where $\tau = x$, or $\tau' = x$, and $n \in \mathsf{N_{L}}$.
		In addition, $x$ is said to be \emph{fresh for $S_{n}$} if $x$ does not occur in $S_{n}$ and $x > y$, for every $y$ that occurs in $S$.
		(These notions can be used with respect to $\mathbf{T}$, whenever $S_{n} \subseteq \mathbf{T}$).
		Without loss of generality, we assume that, whenever $x$ occurs in $S_{n}$, the $n$-labelled constraint $n: \top(x)$ is in $S_{n}$ .
		Also, if $n : r(x, y) \in S_{n}$, we call $y$ an \emph{$r$-successor of $x$} with respect to $S_{n}$.
		Finally, given variables $x, y$ in an $n$-labelled constraint system $S_{n}$, we say that $x$ is \emph{blocked by $y$ in $S_{n}$} if $x > y$ and $\{ C \mid n : C(x) \in S_{n} \} \subseteq \{ C \mid n : C(y) \in S_{n} \}$.

		A completion set $\mathbf{T}$ contains a \emph{clash} if
		%$\{m: \psi, m: \neg\psi\} \subseteq
		%S_{m} \in\mathbf{T}$,
		$\{m: \psi, m: \neg\psi\} \subseteq \mathbf{T}$,
		or
		$\{m: C(x), m: \lnot C(x)\} \subseteq \mathbf{T}$,
		for some $m \in \mathsf{N_{L}}$, and formula $\psi$ or concept $C$.
		A completion set with no clash is \emph{clash-free}.
		%Given a constraint system $S$ for $\p$, we say that $S$ contains a \emph{clash} if there exist a variable $x$ and a concept $C$ such that $\{ x: C, x: \lnot C \} \subseteq S$, or if $\{ \psi, \lnot \psi \} \subseteq S$, for some formula $\psi$.
		Given $\mathbf{L} \in \mathsf{Log}$, a completion set is $\LnALC$-\emph{complete} if no \emph{$\LnALC$-rule} from Figure~\ref{fig:rules} is applicable to $\mathbf{T}$,
		%{\color{red}{
				where $\gamma_{j}$ is either $\psi_{j} \in \forneg(\p)$ or $C_{j}(x_{j})$, with $C_{j} \in \conneg(\p)$, for $j = 1, \ldots, k$, and $\delta$ is either $\chi \in \forneg(\p)$ or $D(y)$, with $D \in \conneg(\p)$,
				%}}
		with respect to the following \emph{application conditions} associated to each $\LnALC$-rule:
		\begin{itemize}
			%\item[$(\mathsf{R}_{\land})$] $\{n :\psi,n : \chi\} \not\subseteq \mathbf{T}$;
			%\item[$(\mathsf{R}_{\lor})$] $\{n :\psi, n : \chi \} \cap \mathbf{T} = \emptyset$;
			%\item[$(\mathsf{R}_{\sqcap})$] $\{n : C(x),n : D(x)\} \not\subseteq \mathbf{T}$;
			%\item[$(\mathsf{R}_{\sqcup})$] $\{n : C(x),n : D(x)\} \cap \mathbf{T} = \emptyset$;
			\item[$(\mathsf{R}_{\land})$] $\{n :\psi,n : \chi\} \not\subseteq \mathbf{T}$;
			\qquad \qquad \ \  $(\mathsf{R}_{\sqcap})$ \ $\{n : C(x),n : D(x)\} \not\subseteq \mathbf{T}$;
			\item[$(\mathsf{R}_{\lor})$] $\{n :\psi, n : \chi \} \cap \mathbf{T} = \emptyset$;
			\qquad \quad $(\mathsf{R}_{\sqcup})$ $\{n : C(x),n : D(x)\} \cap \mathbf{T} = \emptyset$;
			\item[$(\mathsf{R}_{\exists})$] $x$ is not blocked by any variable in $S_{n}$, there is no $z$ such that  $\{n : r(x,z), n : C(z)\} \subseteq \mathbf{T}$, and $y$ is the $<$-minimal
			%\nb{M: added $<$-minimal, ok?}
			variable fresh for $S_{n}$;
			\item[$(\mathsf{R}_{\forall})$] $n : C(y) \notin\mathbf{T}$;
			\item[$(\mathsf{R}_{=})$] $x$ occurs in
			an $n$-labelled constraint
			in $\mathbf{T}$
			%a labelled formula indexed with $\sigma$,
			and $n : C(x) \notin\mathbf{T}$;
			\item[$(\mathsf{R}_{\neq})$] $x$ is
			the $<$-minimal
			%\nb{M: added $<$-minimal, ok?}
			variable 
			fresh
			for $S_{n}$,
			%for $\mathbf{T}$,
			%\nb{T: why not only for $S_{n}$? M: I agree}
			and there is no $y$ such that $n : \dnot C(y) \in \mathbf{T}$;
			\item[$(\mathsf{R}_{\mathbf{L}})$]
			%\textcolor{red}{
				$m$ is fresh
				%\nb{M: add $\ll$-minimal?}
				for $\mathbf{T}$, and there is no $o\in \mathsf{N_{L}}$ such that
				$\{ o: \gamma_1, \ldots, o: \gamma_k, o: \delta\}\subseteq \mathbf{T}$, or $\{ o: \dot{\lnot}\gamma_j, o: \dot{\lnot}\delta\}\subseteq \mathbf{T}$, for some $j\leq l$,
				where $k$ and $l$ are as in Figure~\ref{fig:rules}.
				%\nb{M: written $j\leq k$ instead of $j \leq l$ (cf. case $\mathbf{L} = \mathbf{M}$), ok?}
				%}
		\end{itemize}
		%
		%}}

%4th figure
\begin{figure}
	%\noindent
	%\fbox{\begin{minipage}{\textwidth}
			\centering
			\begin{small}
				%\begin{center}
				\begin{equation*}
					\begin{array}{l l}
						\hbox{\begin{forest}
								for tree={
									calign=center,
									grow'=east, % tree direction
									parent anchor=east, child anchor=west, % edge anchors
								}
								[($\mathsf{R}_{\land}$) \ \ \fbox{ \begin{varwidth}{\textwidth}$n: \psi \land \chi$  \end{varwidth}}
								[ \fbox{ \begin{varwidth}{\textwidth}$n: \psi$ {,} $n:\chi$  \end{varwidth}}]
								]
						\end{forest}}
						&
						\hbox{\begin{forest}
								for tree={
									calign=center,
									grow'=east, % tree direction
									parent anchor=east, child anchor=west, % edge anchors
								}
								[($\mathsf{R}_{\sqcap}$) \ \ \fbox{ \begin{varwidth}{\textwidth}$n: C \sqcap D(x)$\end{varwidth}}
								[ \fbox{ \begin{varwidth}{\textwidth}$n: C(x)$ {,}  $n: D(x)$\end{varwidth}}]
								]
						\end{forest}}
						\\
						\hbox{\begin{forest}
								for tree={
									calign=center,
									grow'=east, % tree direction
									parent anchor=east, child anchor=west, % edge anchors
								}
								[($\mathsf{R}_{\lor}$) \ \ \fbox{ \begin{varwidth}{\textwidth}$n: \psi \lor \chi$ \end{varwidth}}
								[\fbox{ \begin{varwidth}{\textwidth}$n: \psi$\end{varwidth}}]
								[\fbox{ \begin{varwidth}{\textwidth}$n: \chi$\end{varwidth}}]
								]
						\end{forest}}
						&
						\hbox{\begin{forest}
								for tree={
									calign=center,
									grow'=east, % tree direction
									parent anchor=east, child anchor=west, % edge anchors
								}
								[($\mathsf{R}_{\sqcup}$) \ \ \fbox{ \begin{varwidth}{\textwidth}$n: C \sqcup D(x)$ \end{varwidth}}
								[\fbox{ \begin{varwidth}{\textwidth}$n: C(x)$\end{varwidth}}]
								[\fbox{ \begin{varwidth}{\textwidth}$n: D(x)$\end{varwidth}}]
								]
						\end{forest}}
						\\
						\hbox{\begin{forest}
								for tree={
									calign=center,
									grow'=east, % tree direction
									parent anchor=east, child anchor=west, % edge anchors
								}
								[($\mathsf{R}_{\exists}$) \ \ \fbox{ \begin{varwidth}{\textwidth}$n: \exists r.C(x)$ \end{varwidth}}
								[ \fbox{ \begin{varwidth}{\textwidth}{$n: r (x , y)$} {,} $n: C(y)$\end{varwidth}} ]
								]
						\end{forest}}
						&
						\hbox{\begin{forest}
								for tree={
									calign=center,
									grow'=east, % tree direction
									parent anchor=east, child anchor=west, % edge anchors
								}
								[($\mathsf{R}_{\forall}$) \ \ \fbox{ \begin{varwidth}{\textwidth}$n: \forall r.C(x)$  {,}  { $n: r (x , y)$  } \end{varwidth}}
								[ \fbox{ \begin{varwidth}{\textwidth}$n: C(y)$  \end{varwidth}}]
								]
						\end{forest}}
						\\
						\hbox{\begin{forest}
								for tree={
									calign=center,
									grow'=east, % tree direction
									parent anchor=east, child anchor=west, % edge anchors
								}
								[($\mathsf{R}_{=}$) \ \ \fbox{ \begin{varwidth}{\textwidth}$n: \top \sqsubseteq C$  \end{varwidth}}
								[\fbox{ \begin{varwidth}{\textwidth}$n: C(x)$  \end{varwidth}}]
								]
						\end{forest}}
						&
						\hbox{\begin{forest}
								for tree={
									calign=center,
									grow'=east, % tree direction
									parent anchor=east, child anchor=west, % edge anchors
								}
								[($\mathsf{R}_{\neq}$) \ \ \fbox{ \begin{varwidth}{\textwidth}$n: \lnot (\top \sqsubseteq C)$ \end{varwidth}}
								[\fbox{ \begin{varwidth}{\textwidth}$n:\dot{\lnot}C(x)$ \end{varwidth}}]
								]
						\end{forest}}
						\\
					\end{array}
				\end{equation*}
				%\end{center}

				\begin{forest}
					for tree={
						calign=center,
						grow'=east, % tree direction
						parent anchor=east, child anchor=west, % edge anchors
					}
					[($\mathsf{R}_{\mathbf{L}}$) \ \ \fbox{ \begin{varwidth}{\textwidth}$n: \Box_i\gamma_1,\ldots, n: \Box_i\gamma_k, n: \Diamond_{i}\delta$  \end{varwidth}}
					[\fbox{ \begin{minipage}{3.8cm}\centering $ m: \gamma_1, \ldots, m: \gamma_k, m: \delta$  \end{minipage}} \ $(0)$]
					[\fbox{ \begin{minipage}{3.8cm}\centering $ m: \dot{\lnot}\gamma_1, m: \dot{\lnot}\delta$  \end{minipage}} \ $(1)$]
					[\phantom{\qquad\qquad\quad\ \ \ } $\vdots$ \phantom{\qquad\qquad\qquad\quad}]
					[\fbox{ \begin{minipage}{3.8cm}\centering $ m: \dot{\lnot}\gamma_k, m: \dot{\lnot}\delta$  \end{minipage}} \ $(l)$]
					]
				\end{forest}
			\end{small}
			%\end{minipage}}
			\caption{\label{fig:rules} $\LnALC$-rules,
				%for $\mathbf{L} \in \mathsf{Log}$,
				where:
				for $\mathbf{L} = \mathbf{E}$, $k = l = 1$;
				for $\mathbf{L} = \mathbf{M}$, $k = 1$ and $l = 0$;
				for $\mathbf{L} = \mathbf{C}$, $k \geq 1$ and $l = k$;
				for $\mathbf{L} = \mathbf{N}$, $k = l = 1$ or $k = l = 0$;
				%for $\mathbf{L} = \mathbf{MC}$, $k \geq 1$ and $l = 0$;
				%for $\mathbf{L} = \mathbf{MN}$, $l = 0$ and ($k = 1$ or $k = 0$);
				%for $\mathbf{L} = \mathbf{CN}$, ($k \geq 1$ and $l = k$) or $k = l = 0$.
			}
		\end{figure}

		%{\color{blue}{
				The $\LnALC$-rules essentially state how to extend a completion set on the basis of the information contained in it.
				Branching rules entail a non-deterministic choice in the expansion of the completion set.
				For each $\mathbf{L} \in \mathsf{Log}$, 
				we now define an algorithm based on $\LnALC$-rules for checking the $\LnALC$ formula satisfiability.
				We then prove that the algorithm terminates for every formula $\p$, and that it is sound and complete with respect to $\LnALC$ satisfiability.
				
				\begin{definition}[$\LnALC$ tableau algorithm for $\p$]
					Given an $\MLnALC$ formula $\p$, the \emph{$\LnALC$ tableau algorithm for $\p$} runs as follows:
					set the initial completion set $\mathbf{T}_{\p} =  \{0 : \p, 0 : \top(x) \}$,
					and expand %$\mathbf{T}_{\p}$
					it by means of the $\LnALC$-rules  until a clash
					or an $\LnALC$-complete completion set is obtained. % is obtained, or $\mathbf{T}$ is complete.
					%that is, no rule of [TABL] is applicable to $\mathbf{T}$ respecting the application conditions of Definition \ref{def:app cond}.
				\end{definition}

		In the rest of this section, we prove termination, soundness and completeness of the tableau algorithms given above.
		%
		%
		%
		%%% TERMINATION
		%
		We start by showing that the $\LnALC$ tableau algorithm terminates.
		
		%\newpage
		%
		%\subsection*{\textcolor{red}{Tentative Termination}}
		%Let $\mathbf{T}$ be a completion set 
		%constructed
		%expanding the initial completion set $\mathbf{T}_{\p} =  \{0 : \p, 0 : \top(x) \}$ according to the 
		%$\LnALC$ tableau algorithm for $\p$.
		%%\textcolor{red}{We define $\mathsf{L}_{\mathbf{T}} = \{ n \in \mathsf{N_{L}} \mid S_{n} \subseteq \mathbf{T} \}$.}
		%We prove the following two claims.

		\begin{theorem}[Termination]
			\label{thm:termination}
			Having started on the initial completion set $\mathbf{T}_{\p} =  \{0 : \p, 0 : \top(x) \}$, the $\LnALC$ tableau algorithm for $\p$ terminates after at most $2^{p(|\fg(\p)|)}$ steps, where $p$ is a polynomial function. 
		\end{theorem}
		\begin{proof}
			
			We first require the following claims.
			%The statement is a consequence of the following claims.
			
			\begin{claim}
				\label{cla:termlocal}
				Let $\mathbf{T}$ be a completion set obtained by applying the $\LnALC$ tableau algorithm for $\p$.
				For each $n \in \mathsf{L}_{\mathbf{T}}$,
				%Let $S_{n}$ be a $n$-labelled constraint system for $\p$ in 
				%$\mathbf{T}$.
				the number of 
				%constraints
				$n$-labelled constraints
				%of the form $n: \psi$, or $n: C(x)$, or $n: r(x, y)$ 
				for $\p$
				in $\mathbf{T}$ does not exceed $2^{q(|\fg(\p)|)}$, where $q$ is a polynomial function. 
			\end{claim}
			\begin{proof}[Proof of Claim]
				%\nb{M: changed, to be checked}
				We remark that, for each $S_n\subseteq\mathbf T$, the $\LnALC$ tableaux algorithm behaves exactly like a standard (non-modal) $\ALC$ tableaux algorithm (cf. e.g.~\cite[Theorem 15.4]{GabEtAl03}, noting also that in our case we do not have to deal with individual names).
			\end{proof}
			
			\begin{claim}
				\label{cla:termglobal}
				Let $\mathbf{T}$ be a completion set obtained by applying the $\LnALC$ tableau algorithm for $\p$.
				For
				$\mathbf{L} \in \{\mathbf{E}, \mathbf{M}, \mathbf{N}\}$,
				%when started on $\mathbf{T}_{\p}$,
				%the number of labels generated by the $\LnALC$ tableau algorithm does not exceed $|\fg(\p)|^2$.
				$|\mathsf{L}_{\mathbf{T}}| \leq |\fg(\p)|^2$.
				%%the cardinality of $\mathsf{L}_{\mathbf{T}}$
				%%does not exceed
				%%$|\fg(\p)|^2$.
				For
				$\mathbf{L} = \mathbf{C}$,
				%when started on $\mathbf{T}_{\p}$,
				%the number of labels generated by the $\LnALC$ tableau algorithm does not exceed
				%$2^{|\fg(\p)|} \cdot |\fg(\p)|$.
				$|\mathsf{L}_{\mathbf{T}}| \leq 2^{|\fg(\p)|} \cdot |\fg(\p)|$.
				%%the cardinality of $\mathsf{L}_{\mathbf{T}}$
				%%does not exceed
				%%$2^{|\fg(\p)|} \cdot |\fg(\p)|$.
			\end{claim}
			\begin{proof}[Proof of Claim]
				%World
				Labels $n$ are generated in $\mathbf{T}$ by means
				of the application of the rule $\mathsf{R}_{\mathbf{L}}$.
				For $\mathbf{L} \in \{\mathbf{E}, \mathbf{M}, \mathbf{N}\}$,
				this rule is applied to two $n$-labelled contraints
				$n: \Box_i\gamma, n: \Diamond_{i}\delta$
				(for $\mathbf{L} = \mathbf{N}$ possibly also to a single constraint
				$n: \Diamond_{i}\delta$),
				%whereas 
				while for $\mathbf{L} = \mathbf{C}$
				it is applied to $k+1$ $n$-labelled contraints
				$n: \Box_i\gamma_1, ... n: \Box_i\gamma_{k}, n: \Diamond_{i}\delta$.
				By the application condition of $\mathsf{R}_{\mathbf{L}}$,
				each such combination
				of
				%\textcolor{red}{modal
					constraints
					%}
				generates at most one label $m$.
				Therefore, the number of
				%world
				labels that can be generated in $\mathbf{T}$ is bounded by the number of possible
				such
				combinations,
				% of this kind,
				which is at most $|\fg(\p)|^2$, for $\mathbf{L} \in \{\mathbf{E}, \mathbf{M}, \mathbf{N}\}$,
				and at most $2^{|\fg(\p)|} \cdot |\fg(\p)|$, for $\mathbf{L} = \mathbf{C}$.
				%\nb{T: which symbol for multiplication? M: this one you used}
			\end{proof}
			
			The theorem is then a consequence of the following observations.
			Given a completion set $\mathbf{T}$ constructed by the $\LnALC$ tableau algorithm,
			we have by Claim~\ref{cla:termglobal} that
			the number of applications of rule
			$\mathsf{R}_{\mathbf{L}}$
			is bounded by $| \mathsf{L}_{\mathbf{T}} |$, which is at most 
			$|\fg(\p)|^2$,
			for
			$\mathbf{L} \in \{\mathbf{E}, \mathbf{M}, \mathbf{N}\}$,
			and at most
			$2^{|\fg(\p)|} \cdot |\fg(\p)|$,
			for
			$\mathbf{L} = \mathbf{C}$.
			Moreover, since every application of the rules
			$\mathsf{R}_{\land}$ and $\mathsf{R}_{\lor}$ introduces a new formula to an $n$-labelled constraint, the total number of such rule applications is bounded by $| \mathsf{L}_{\mathbf{T}} | \cdot | \fg(\p) |$.
			Finally, by Claim~\ref{cla:termlocal}, the number of applications of rules
			$\mathsf{R}_{\sqcap}, \mathsf{R}_{\sqcup}, \mathsf{R}_{\forall}, \mathsf{R}_{\exists}, \mathsf{R}_{=}, \mathsf{R}_{\neq}$ per label $n$ is bounded by $2^{q(|\fg(\p)|)}$, where $q$ is a polynomial function, since these rules add a new constraint to an $n$-labelled constraint system.
			Thus, the overall number of such rule applications is bounded by $| \mathsf{L}_{\mathbf{T}} | \cdot 2^{q(|\fg(\p)|)}$.
		\end{proof}

		We now proceed to prove that the $\LnALC$ tableau algorithm is sound.
		
		%\newpage
		\begin{theorem}[Soundness]
			\label{thm:soundness}
			If, having started on the initial completion set $\mathbf{T}_{\p}$, the $\LnALC$ tableau algorithm constructs an $\LnALC$-complete and clash-free completion set for $\p$, then $\p$ is $\LnALC$ satisfiable.
		\end{theorem}
		\begin{proof}
			Given an $\LnALC$-complete and clash-free completion set $\mathbf{T}$ for $\p$,
			%let $\mathsf{L}_{\mathbf{T}} = \{ n \in \mathsf{N_{L}} \mid S_{n} \subseteq \mathbf{T} \}$.
			define, for $n \in \mathsf{L}_{\mathbf{T}}$, $\psi \in \forneg(\p)$, $C \in \conneg(\p)$, and $x$ occurring in $\mathbf{T}$,
			%\begin{align*}
			%	\lfloor \psi \rfloor & = \{ n \in \mathsf{L}_{\mathbf{T}} \mid n : \psi \in S_{n} \}, \\
			%	\lfloor C \rfloor_{x} & = \{ n \in \mathsf{L}_{\mathbf{T}} \mid n : C(x) \in S_{n} \},
			%\end{align*}
			%and
			%\begin{align*}
			%	\lceil \psi \rceil & = \mathsf{L}_{\mathbf{T}} \setminus \{ n \in \mathsf{L}_{\mathbf{T}} \mid n : \dnot\psi \in S_{n}\}, \\
			%	\lceil C \rceil_{x} & = \mathsf{L}_{\mathbf{T}} \setminus \{ n \in \mathsf{L}_{\mathbf{T}} \mid n : \dnot C(x) \in S_{n} \}.
			%\end{align*}
			
			\begin{center}
				\begin{tabular}{lll}	
					\vspace{0.25cm}
					$\lfloor C \rfloor_{x}$ = $\{ n \in \mathsf{L}_{\mathbf{T}} \mid n : C(x) \in S_{n} \}$, & \ \ &
					$\lfloor \psi \rfloor$ = $\{ n \in \mathsf{L}_{\mathbf{T}} \mid n : \psi \in S_{n} \}$, \\
					\vspace{0cm}
					$\lceil C \rceil_{x}$ = $\mathsf{L}_{\mathbf{T}} \setminus \{ n \in \mathsf{L}_{\mathbf{T}} \mid n : \dnot C(x) \in S_{n} \}$, &&
					$\lceil \psi \rceil$ = $\mathsf{L}_{\mathbf{T}} \setminus \{ n \in \mathsf{L}_{\mathbf{T}} \mid n : \dnot\psi \in S_{n}\}$. \\
				\end{tabular}
			\end{center}
			
			\noindent
			Moreover, define $\Gamma^{x}_{n} =  \{ \psi \mid n : \psi \in S_{n} \} \cup \{ C \mid  n : C(x) \in S_{n} \}$ and let $\gamma, \delta$ range over $\MLnALC$ formulas or concepts,
			where: $\lfloor \gamma \rfloor_{x} = \lfloor \psi \rfloor$, if $\gamma = \psi$, and $\lfloor \gamma \rfloor_{x}  = \lfloor C \rfloor_{x} $, if $\gamma = C$; and similarly for $\lceil \gamma \rceil_{x}$.
			%so that: if $\gamma = C$ and $\Box_{i}\gamma \in \Gamma^{x}_{n} $, then $\lceil \gamma \rceil = \lceil C \rceil_{x}$, $\rfloor \gamma \lfloor = \rfloor C \lfloor_{x}$; whereas, if $\gamma = \psi$ and $\Box_{i}\gamma \in \Gamma^{x}_{n} $< then $\lceil \gamma \rceil = \lceil \psi \rceil$, $\rfloor \gamma \lfloor = \rfloor \psi \lfloor$.
			%
			We set $\Mmc = (\Fmc, \Imc)$, with $\Fmc = (\Wmc, \{ \Nmc_{i} \}_{i \in I})$ and $\Imc_{n} = (\Delta_{n}, \cdot^{\Imc_{n}})$, for $n \in \Wmc$, defined as follows:
			\begin{itemize}
				\item $\Wmc =  \mathsf{L}_{\mathbf{T}}$;
				\item for every $i \in I = \{1, \ldots, n\}$, we set $\Nmc_{i} \colon \W \rightarrow 2^{2^{\Wmc}}$ such that:

								\begin{center}
									\begin{tabular}{ll}
										\vspace{0.2cm}
										-- \ for $\mathbf{L} = \mathbf{E}$: & 
										$\Nmc_{i}(n) = \big\{ \alpha \mid \textnormal{for some} \ \Box_{i}\gamma \in \Gamma^{x}_{n}  \colon \lfloor \gamma \rfloor_{x} \subseteq \alpha \subseteq \lceil \gamma \rceil_{x} \big\}$; \\
										
										\vspace{0.2cm}	 
										-- \ for $\mathbf{L} = \mathbf{M}$: &
										$\Nmc_{i}(n) = \big\{ \alpha \mid \textnormal{for some} \  \Box_{i}\gamma \in \Gamma^{x}_{n}  \colon \lfloor \gamma \rfloor_{x} \subseteq \alpha \big\}$; \\
										
										\vspace{0cm}			
										-- \ for $\mathbf{L} = \mathbf{C}$: &
										$\Nmc_{i}(n) = \big\{ \alpha \mid \textnormal{for some} \ \Box_{i}\gamma_{1}  \in \Gamma^{{x}_{1}}_{n}, \ldots, \Box_{i}\gamma_{k} \in \Gamma^{{x}_{k}}_{n} \colon$ \\
										\vspace{0.2cm}
										& \phantom{$\Nmc_{i}(n) = \{  \alpha \mid \, $}
										$\bigcap^{k}_{j = 1} \lfloor \gamma_{j} \rfloor_{{{x}_{j}}}
										\subseteq \alpha \subseteq
										\bigcap^{k}_{j = 1} \lceil \gamma_{j} \rceil_{{{x}_{j}}} \big\}$; \\
										
										\vspace{0cm}		
										-- \ for $\mathbf{L} = \mathbf{N}$: & 
										$\Nmc_{i}(n) = \big\{ \alpha \mid \textnormal{for some} \ \Box_{i}\gamma \in \Gamma^{x}_{n}  \colon \lfloor \gamma \rfloor_{x} \subseteq \alpha \subseteq \lceil \gamma \rceil_{x} \big\} \cup \Wmc$; \\
									\end{tabular}
								\end{center}
								
								\item $\Delta_{n} = \{ x \in \mathsf{N_{V}} \mid x \ \text{occurs in} \ S_{n} \}$;
								\item $A^{\Imc_{n}} = \{ x \in \Delta_{n} \mid n : A(x) \in S_{n} \}$;
								\item $r^{\Imc_{n}} = \{ (x, y) \in \Delta_{n} \times \Delta_{n} \mid n : r(x, y) \in S_{n} \ \text{or} \ n : r(z, y) \in S_{n}, $
								for some $z$ blocking $x$ in $S_{n}$ \}.
							\end{itemize}
							
							First, we observe the following.
							
							\begin{itemize}
								\item For $\mathbf{L}  = \mathbf{M}$, we have that $\Mmc = (\Fmc, \Int)$ is such that $\Fmc = (\Wmc,  \{\Nmc_i \}_{i \in I})$ is supplemented. Indeed,
								for all $n \in \Wmc$, $\alpha,\beta\subseteq \Wmc$, suppose that $\alpha\in \Nmc_{i}(n)$ and $\alpha \subseteq \beta$. By definition, this implies that: for some $\Box_{i} \gamma \in \Gamma^{x}_{n} $, $\lfloor \gamma \rfloor_{x} \subseteq \alpha \subseteq \beta$. Hence, $\beta \in \Nmc_{i}(n)$.
								\item For $\mathbf{L}  = \mathbf{C}$, we have that $\Mmc = (\Fmc, \Int)$ is such that $\Fmc = (\Wmc,  \{\Nmc_i \}_{i \in I})$ is closed under intersection. Indeed, for all $n \in \Wmc$, $\alpha,\beta\subseteq \Wmc$, suppose that $\alpha\in \Nmc_{i}(n)$ and $\beta\in \Nmc_{i}(n)$.
								Now suppose that, for some
								$\Box_{i}\gamma_{1}  \in \Gamma^{{x}_{1}}_{n}, \ldots, \Box_{i}\gamma_{k} \in \Gamma^{{x}_{k}}_{n} \colon
								\bigcap^{k}_{j = 1} \lfloor \gamma_{j} \rfloor_{{x}_{j}}
								\subseteq \alpha \subseteq
								\bigcap^{k}_{j = 1} \lceil \gamma_{j} \rceil_{{x}_{j}}$
								and, for some
								$\Box_{i}\delta_{1} \in \Gamma^{y_{1}}_{n}, \ldots, \Box_{i}\delta_{h} \in \Gamma^{y_{h}}_{n}  \colon
								\bigcap^{h}_{j = 1} \lfloor \delta_{j} \rfloor_{y_{j}}
								\subseteq \beta \subseteq
								\bigcap^{h}_{j = 1} \lceil \delta_{j} \rceil_{y_{j}}$.
								%				 This implies that, for some $\Box_{i}\gamma_{1}  \in \Gamma^{{x}_{1}}_{n}, \ldots, \Box_{i}\gamma_{k} \in \Gamma^{{x}_{k}}_{n}$ and some $\Box_{i}\delta_{1} \in \Gamma^{y_{1}}_{n}, \ldots, \Box_{i}\delta_{h} \in \Gamma^{y_{h}}_{n} $, we have
								%				 \[
								%				 \bigcap^{k}_{j = 1} \lfloor \gamma_{j} \rfloor_{{x}_{j}} \cap \bigcap^{h}_{j = 1} \lfloor \delta_{j} \rfloor_{y_{j}}
								%				 \subseteq
								%				 \alpha \cap \beta
								%				 \subseteq
								%				 \bigcap^{k}_{j = 1} \lceil \gamma_{j} \rceil_{{x}_{j}} \cap \bigcap^{h}_{j = 1} \lceil \delta_{j} \rceil_{y_{j}}
								%				 \]
								%	This implies that $\alpha\cap\beta\in \Nmc_{i}(n)$.
								Then for some $\Box_{i}\gamma_{1}  \in \Gamma^{{x}_{1}}_{n}, \ldots, \Box_{i}\gamma_{k} \in \Gamma^{{x}_{k}}_{n}$ and some $\Box_{i}\delta_{1} \in \Gamma^{y_{1}}_{n}, \ldots, \Box_{i}\delta_{h} \in \Gamma^{y_{h}}_{n} $ the following holds, which in turn implies that
								$\alpha\cap\beta\in \Nmc_{i}(n)$:
								%				 \[
								%				 \bigcap^{k}_{j = 1} \lfloor \gamma_{j} \rfloor_{{x}_{j}} \cap \bigcap^{h}_{j = 1} \lfloor \delta_{j} \rfloor_{y_{j}}
								%				 \subseteq
								%				 \alpha \cap \beta
								%				 \subseteq
								%				 \bigcap^{k}_{j = 1} \lceil \gamma_{j} \rceil_{{x}_{j}} \cap \bigcap^{h}_{j = 1} \lceil \delta_{j} \rceil_{y_{j}}
								%				 \]				
								\begin{center}$\bigcap^{k}_{j = 1} \lfloor \gamma_{j} \rfloor_{{x}_{j}} \cap \bigcap^{h}_{j = 1} \lfloor \delta_{j} \rfloor_{y_{j}}
									\subseteq
									\alpha \cap \beta
									\subseteq
									\bigcap^{k}_{j = 1} \lceil \gamma_{j} \rceil_{{x}_{j}} \cap \bigcap^{h}_{j = 1} \lceil \delta_{j} \rceil_{y_{j}}$\end{center}
								
								\item For $\mathbf{L}  = \mathbf{N}$, we have that $\Mmc = (\Fmc, \Int)$, with $\Fmc = (\Wmc,  \{\Nmc_i \}_{i \in I})$, is such that $\Fmc$ contains the unit. Indeed, by construction, for all $n \in \Wmc$, $\Wmc \in \Nmc_{i}(n)$.
							\end{itemize}
							
							We then require the following claims.
							\begin{claim}
								\label{cla:conind}
								For every $n \in \Wmc$, $C \in \conneg(\p)$, and $x \in \Delta_{n}$: if $n : C(x) \in S_{n}$, then $x \in C^{\Imc_{n}}$.
							\end{claim}
							\begin{proof}[Proof of Claim]
								We show the claim by induction on the weight of $C$ (in NNF).
								The base case of $C = A$ comes immediately from the definitions.
								For the base case of $C = \lnot A$, suppose that $n : \lnot A(x) \in S_{n}$. Since $\mathbf{T}$ is clash-free, we have that $n : A(x) \not \in S_{n}$, and thus $x \not \in A^{\Imc_{n}}$ by definition of $A^{\Imc_{n}}$, meaning $x \in (\lnot A)^{\Imc_{n}}$.
								The inductive cases of $C = D \sqcap E$ and $C = D \sqcup E$ come from the fact that $S_{n}$ is closed under $\mathsf{R}_{\sqcap}$ and $\mathsf{R}_{\sqcup}$, respectively, and straightforward applications of the inductive hypothesis.
								We show the remaining cases (cf. also~\cite[Claim 15.2]{GabEtAl03}).
								
								%\nb{M: added, to be checked}
								$C = \exists r.D$.
								Let $n : \exists r.D(x) \in S_{n}$, meaning that $\exists r.D \in \Gamma^{x}_{n}$. We distinguish two cases.
								\begin{itemize}
									\item $x$ is not blocked by any variable in $S_{n}$. Since $S_{n}$ is closed under $\mathsf{R}_{\exists}$, there exists $y$ occurring in $S_{n}$ such that $n : r(x,y) \in S_{n}$ and $n : D(y) \in S_{n}$. Thus, by definition, $(x, y) \in r^{\Imc_{n}}$ and $n : D(y) \in S_{n}$. By inductive hypothesis, we obtain that $x \in (\exists r.D)^{\Imc_{n}}$.
									\item $x$ is blocked by a variable in $S_{n}$, implying that there exists a $<$-minimal (since $<$ is a well-ordering) $y$ occurring in $S_{n}$ such that $y < x$ and $\{ E \mid n : E(x) \in S_{n} \} \subseteq \{ E \mid n : E(y) \in S_{n} \}$.
									In turn, this implies that $y$ is not blocked by any other variable $z$ in $S_{n}$ (for otherwise $z$ would block $x$, with $z < y$, against the fact that $y$ is $<$-minimal).
									By reasoning as in the case above, since $y$ is not blocked and $S_{n}$ is closed under $\mathsf{R}_{\exists}$, we have a variable $z$ occurring in $S_{n}$ such that $n : r(y,z) \in S_{n}$ and $n : D(x) \in S_{n}$.
									Since $y$ blocks $x$, by definition we have that $(x, z) \in r^{\Imc_{n}}$, and by inductive hypothesis we get from $n : D(z)$ that $z \in D^{\Imc_{n}}$.
									Thus, $x \in (\exists r.D)^{\Imc_{n}}$.
								\end{itemize}
								
								$C = \forall r.D$. 
								Let $n : \forall r.D(x) \in S_{n}$, meaning that $\forall r.D \in \Gamma^{x}_{n}$, and suppose that $(x, y) \in r^{\Imc_{n}}$. By definition, either $n : r(x,y) \in S_{n}$ or $n : r(z,y) \in S_{n}$, for some $z$ blocking $x$ in $S_{n}$.
								In the former case, since $S_{n}$ is closed under $\mathsf{R}_{\forall}$, we get that $n : D(y) \in S_{n}$.
								In the latter case, since $z$ blocks $x$ in $S_{n}$, we obtain $n : \forall r.D(z) \in S_{n}$; again, since $S_{n}$ is closed under $\mathsf{R}_{\forall}$, this implies that $n : D(y) \in S_{n}$.
								Hence, in both cases, we have $n : D(y) \in S_{n}$.
								By inductive hypothesis, this means that $y \in D^{\Imc_{n}}$.
								Since $y$ was arbitrary, we conclude that $x \in (\forall r.D)^{\Imc_{n}}$.
								
								%The inductive cases of $C = \exists r.D$ and $C = \forall r.D$ can be proved analogously to~\cite[Claim 15.2]{GabEtAl03}.\nb{M: todo add?}
								%We show the %remaining
								%modal cases.
								
								%$C = \exists r.D$. \ldots\nb{M: todo add}
								%
								%$C = \forall r.D$. \ldots\nb{M: todo add}
								
								$C = \Box_{i} D$.
								Let $n : \Box_{i} D(x) \in S_{n}$, meaning that $\Box_{i} D \in \Gamma^{x}_{n}$.
								Consider
								$\mathbf{L} \in \mathsf{Log}$.
								%$\mathbf{L} \in \{ \mathbf{E},  \mathbf{M},  \mathbf{C},  \mathbf{N} \}$.
								\begin{enumerate}[leftmargin=*, align=left]
									\item[$\mathbf{L} = \mathbf{E}$.]
									%	Let $\mathbf{L} = \mathbf{E}$.
									We have by inductive hypothesis that
									$\lfloor D \rfloor_{x} = \{ n \in \Wmc \mid n : D(x) \in S_{n} \} \subseteq \{ n \in \Wmc \mid x \in D^{\Imc_{n}} \} = \llbracket D \rrbracket^{\Mmc}_{x}$.
									By inductive hypothesis %(since $| D | = | \dnot D |$), 
									(since $| \dnot D | = | D |$), 
									we also have that
									$\{ n \in \Wmc \mid n : \dnot D(x) \in S_{n} \} \subseteq \{ n \in \Wmc \mid x \in (\dnot D)^{\Imc_{n}} \} = \llbracket \dnot D \rrbracket^{\Mmc}_{x} = \Wmc \setminus \llbracket D \rrbracket^{\Mmc}_{x}$.
									Hence, $\llbracket D \rrbracket^{\Mmc}_{x} \subseteq \Wmc \setminus \{ w \in \Wmc \mid n : \dnot D(x) \in S_{n} \} = \lceil D \rceil_{x}$. In conclusion, we have $\Box_{i} D \in \Gamma^{x}_{n} $ such that $\lfloor D \rfloor_{x} \subseteq \llbracket D \rrbracket^{\Mmc}_{x} \subseteq \lceil D \rceil_{x}$. Thus, by definition, $\llbracket D \rrbracket^{\Mmc}_{x} \in \Nmc_{i}(n)$, as required.
									\item[$\mathbf{L} = \mathbf{M}$.]
									We have by inductive hypothesis that
									$\lfloor D \rfloor_{x} = \{ n \in \Wmc \mid n : D(x) \in S_{n} \} \subseteq \{ n \in \Wmc \mid x \in D^{\Imc_{n}} \} = \llbracket D \rrbracket^{\Mmc}_{x}$.
									Thus, we have $\Box_{i} D \in \Gamma^{x}_{n} $ such that $\lfloor D \rfloor_{x} \subseteq \llbracket D \rrbracket^{\Mmc}_{x}$. By definition, this means $\llbracket D \rrbracket^{\Mmc}_{x} \in \Nmc_{i}(n)$, as required.

									%	\item[$\mathbf{L} = \mathbf{C}$.] This case is analogous to the case for $\mathbf{L} = \mathbf{E}$.
									%	
									%	\item[$\mathbf{L} = \mathbf{N}$.] This case is analogous to the case for $\mathbf{L} = \mathbf{E}$.
									\item[$\mathbf{L} \in\{ \mathbf{C}, \mathbf{N}\}$.] This cases are analogous to the case for $\mathbf{L} = \mathbf{E}$.
								\end{enumerate}
								%
								%The cases of $\mathbf{L} \in \{ \mathbf{M}, \mathbf{C}, \mathbf{N}\}$ are proved analogously.\nb{M: todo add}

								$C = \Diamond_{i} D$. Let $n : \Diamond_{i}D(x) \in S_{n}$. Consider 
								%$\mathbf{L} \in \{ \mathbf{E},  \mathbf{M},  \mathbf{C},  \mathbf{N} \}$.
								$\mathbf{L} \in \mathsf{Log}$.
								\begin{enumerate}[leftmargin=*, align=left]
									%	\item[$\mathbf{L} = \mathbf{E}$.]
									%	We distinguish two cases.
									%	\begin{itemize}
										%		\item There exists no $\Box_{i} \gamma \in \Gamma^{y}_{n}$. This means that $\Nmc_{i}(n) = \emptyset$. Thus, $\Wmc \setminus \llbracket D \rrbracket^{\Mmc}_{x} \not \in \Nmc_{i}(n)$, meaning that $x \in (\Diamond_{i}D)^{\Imc_{n}}$.
										%		\item There exists $\Box_{i} \gamma \in \Gamma^{y}_{n}$. We then reason similarly to the case for $\mathbf{L} = \mathbf{C}$.
										%	\end{itemize}
									%	
									%	\item[$\mathbf{L} = \mathbf{M}$.]
									%	We distinguish two cases.
									%	\begin{itemize}
										%		\item There exists no $\Box_{i} \gamma \in \Gamma^{y}_{n}$. As for $\mathbf{L} = \mathbf{E}$, this implies that $x \in (\Diamond_{i}D)^{\Imc_{n}}$.
										%		\item There exists $\Box_{i} \gamma \in \Gamma^{y}_{n}$. We then reason similarly to the case for $\mathbf{L} = \mathbf{C}$.
										%	\end{itemize}
									\item[$\mathbf{L} \in \{\mathbf{E}, \mathbf{M}\}$.]
									We distinguish two cases.
									%	\begin{itemize}
										%		\item There exists no $\Box_{i} \gamma \in \Gamma^{y}_{n}$. This means that $\Nmc_{i}(n) = \emptyset$. Thus, $\Wmc \setminus \llbracket D \rrbracket^{\Mmc}_{x} \not \in \Nmc_{i}(n)$, meaning that $x \in (\Diamond_{i}D)^{\Imc_{n}}$.
										%		\item There exists $\Box_{i} \gamma \in \Gamma^{y}_{n}$. We then reason similarly to the case for $\mathbf{L} = \mathbf{C}$.
										%	\end{itemize}
									$(i)$ There exists no $\Box_{i} \gamma \in \Gamma^{y}_{n}$. This means that $\Nmc_{i}(n) = \emptyset$. Thus, $\Wmc \setminus \llbracket D \rrbracket^{\Mmc}_{x} \not \in \Nmc_{i}(n)$, meaning that $x \in (\Diamond_{i}D)^{\Imc_{n}}$.
									$(ii)$ There exists $\Box_{i} \gamma \in \Gamma^{y}_{n}$. We then reason similarly to the case for $\mathbf{L} = \mathbf{C}$.

									\item[$\mathbf{L} = \mathbf{C}$.]
									We distinguish two cases.
									$(i)$ There exist no $\Box_{i} \gamma_{1} \in \Gamma^{y_{1}}_{n}, \ldots, \Box_{i} \gamma_{k} \in \Gamma^{y_{k}}_{n}$. As for $\mathbf{L} = \mathbf{E}$, we obtain $x \in (\Diamond_{i}D)^{\Imc_{n}}$.
									$(ii)$ There exist $\Box_{i} \gamma_{1} \in \Gamma^{y_{1}}_{n}, \ldots, \Box_{i} \gamma_{k} \in \Gamma^{y_{k}}_{n}$.
									Since $\mathbf{T}$ is $\LnALC$-complete, there exists $m \in \Wmc$ such that:
									%				\begin{enumerate}[label=$(\arabic*)$, start=0]
										%					\item $\gamma_{1}  \in \Gamma^{y_{1}}_{m}, \ldots, \gamma_{k} \in \Gamma^{y_{k}}_{m}$ and $D \in \Gamma^{x}_{m}$; or
										%					\item $\dnot \gamma_{1} \in \Gamma^{y_{1}}_{m}$ and $\dnot D \in \Gamma^{x}_{m}$; or
										%					\item[] $\vdots$
										%					\item[$(l)$] $\dnot \gamma_{l} \in \Gamma^{y_{l}}_{m}$ and $\dnot D \in \Gamma^{x}_{m}$.
										%				\end{enumerate}
									$\gamma_{1}  \in \Gamma^{y_{1}}_{m}, \ldots, \gamma_{k} \in \Gamma^{y_{k}}_{m}$ and $D \in \Gamma^{x}_{m}$; or
									$\dnot \gamma_{j} \in \Gamma^{y_{j}}_{m}$ and $\dnot D \in \Gamma^{x}_{m}$, for some $j\leq k$.
									By inductive hypothesis, the previous step implies that there exists $m \in \Wmc$ such that:
									%				\begin{enumerate}[label=$(\arabic*)$, start=0]
										%					\item $\gamma_{1} \in \Gamma^{y_{1}}_{m}, \ldots, \gamma_{k} \in \Gamma^{y_{k}}_{m}$ and $x \in D^{\Imc_{m}}$; or
										%					\item $\dnot \gamma_{1} \in \Gamma^{y_{1}}_{m}$ and $x \in \dnot D^{\Imc_{m}}$; or
										%					\item[] $\vdots$
										%					\item[$(l)$] $\dnot \gamma_{l} \in \Gamma^{y_{l}}_{m}$ and $x \in \dnot D^{\Imc_{m}}$.
										%				\end{enumerate}
									$\gamma_{1} \in \Gamma^{y_{1}}_{m}, \ldots, \gamma_{k} \in \Gamma^{y_{k}}_{m}$ and $x \in D^{\Imc_{m}}$; or
									$\dnot \gamma_{j} \in \Gamma^{y_{j}}_{m}$ and $x \in \dnot D^{\Imc_{m}}$, for some $j\leq k$.
									Equivalently, it is not the case that, for every $v \in \Wmc$:
									%				\begin{enumerate}[label=$(\arabic*)$, start=0]
										%					\item $\gamma_{1} \in \Gamma^{y_{1}}_{m}, \ldots, \gamma_{k} \in \Gamma^{y_{k}}_{m}$ implies $x \not \in D^{\Imc_{m}}$; and
										%					\item $x \in \dnot D^{\Imc_{m}}$ implies $\dnot \gamma_{1} \not \in \Gamma^{y_{1}}_{m}$; and
										%					\item[] $\vdots$
										%					\item[$(l)$] $x \in \dnot D^{\Imc_{m}}$ implies $\dnot \gamma_{k} \not \in \Gamma^{y_{l}}_{m}$.
										%				\end{enumerate}
									$\gamma_{1} \in \Gamma^{y_{1}}_{m}, \ldots, \gamma_{k} \in \Gamma^{y_{k}}_{m}$ implies $x \not \in D^{\Imc_{m}}$; and
									for all $j\leq k$, $x \in \dnot D^{\Imc_{m}}$ implies $\dnot \gamma_{j} \not \in \Gamma^{y_{j}}_{m}$.
									In other words, it is not the case that:
									%				\begin{enumerate}[label=$(\arabic*)$, start=0]
										%					\item $\bigcap_{j = 1}^{k} \lfloor \gamma_{j} \rfloor_{y_{j}} \subseteq \Wmc \setminus \llbracket D \rrbracket^{\Mmc}_{x}$; and
										%					\item $\Wmc \setminus \llbracket D \rrbracket^{\Mmc}_{x} \subseteq \lceil \gamma_{1} \rceil_{y_{1}}$; and
										%					\item[] $\vdots$
										%					\item[$(l)$]  $\Wmc \setminus \llbracket D \rrbracket^{\Mmc}_{x} \subseteq \lceil \gamma_{l} \rceil_{y_{l}}$.
										%				\end{enumerate}
									$\bigcap_{j = 1}^{k} \lfloor \gamma_{j} \rfloor_{y_{j}} \subseteq \Wmc \setminus \llbracket D \rrbracket^{\Mmc}_{x}$; and
									$\Wmc \setminus \llbracket D \rrbracket^{\Mmc}_{x} \subseteq \bigcap_{j = 1}^{k}\lceil \gamma_{l} \rceil_{y_{l}}$.
									Thus, $\Wmc \setminus  \llbracket D \rrbracket^{\Mmc}_{x} \not \in \Nmc_{i}(n)$, i.e., $x \in (\Diamond_{i}D)^{\Imc_{n}}$, as required.

									\item[$\mathbf{L} = \mathbf{N}$.]
									We distinguish two cases.
									$(i)$ There exists no $\Box_{i} \gamma \in \Gamma^{y}_{n} $. This means that $\Nmc_{i}(n) = \Wmc$.
									Since $\mathbf{T}$ is $\LnALC$-complete, there exists $m \in \Wmc$ such that $D \in \Gamma^{x}_{m}$, i.e., $m : D(x) \in S_{m}$. By inductive hypothesis, this implies $x \in D^{\Imc_{m}}$, that is, $\llbracket D \rrbracket^{\Mmc}_{x} \neq \emptyset$. This holds iff $\Wmc \setminus \llbracket D \rrbracket^{\Mmc}_{x} \neq \Wmc$, and thus $\Wmc \setminus \llbracket D \rrbracket^{\Mmc}_{x} \not \in \Nmc_{i}(n)$. Hence, $x \in (\Diamond_{i}D)^{\Imc_{n}}$.
									%		Thus, $\Wmc \setminus \llbracket D \rrbracket^{\Mmc}_{x} \not \in \Nmc_{i}(n)$ iff $\llbracket D \rrbracket^{\Mmc}_{x} \neq \emptyset$, which in turn is equivalent to $D \neq \bot$. This last step holds, for otherwise there would exist $v \in \Wmc$ such that $v: \bot(\rho(v)) \in S_{m}$, contrary to the fact that $S_{m}$ is clash-free by construction. Hence, $\rho \in (\Diamond_{i}D)^{\Imc_{n}}$.
									$(ii)$ There exists $\Box_{i} \gamma \in \Gamma^{y}_{n}$. We then reason similarly to the case for $\mathbf{L} = \mathbf{C}$.	\qedhere
								\end{enumerate}
							\end{proof}

							\begin{claim}
								\label{cla:forind}
								For every $w \in \Wmc$ and $\psi \in \conneg(\p)$: if $n : \psi \in S_{n}$, then $\Mmc, w \models \psi$.
							\end{claim}
							\begin{proof}[Proof of Claim]
								We prove the claim by induction on the weight of $\p$ (in NNF).
								
								$\psi = (\top \sqsubseteq C)$. Let $n : \top \sqsubseteq C \in S_{n}$ and let $x \in \Delta_{n}$. Since $S_{n}$ is closed under $(\mathsf{R}_{=})$ and $x$ occurs in $S_{n}$, we have that $n : C(x) \in S_{n}$. By Claim~\ref{cla:conind}, we have that $x \in C^{\Imc_{n}}$. Given that $x$ is arbitrary, we conclude that $\Mmc, n \models \top \sqsubseteq C$.
								
								$\psi = \lnot (\top \sqsubseteq C)$. Let $n : \lnot (\top \sqsubseteq C) \in S_{n}$. Since $S_{n}$ is closed under $(\mathsf{R}_{\neq})$, there exists $x$ occurring in $S_{n}$ such that $n : \dnot C(x) \in S_{n}$. By Claim~\ref{cla:conind}, we obtain that $x \in (\dnot C)^{\Imc_{n}}$, for some $x \in \Delta_{w}$. Hence, $\Mmc, n \models \lnot (\top \sqsubseteq C)$.
								
								The inductive cases of 
								$\psi = \chi \land \vartheta$
								and
								$\psi = \chi \lor \vartheta$ follow from the definitions and straighforward applications of the inductive hypothesis.
								%
								%The inductive cases of 
								Moreover the inductive cases of 
								$\psi = \Box_{i} \chi$
								and 
								$\psi = \Diamond_{i} \chi$ can be proved analogously to Claim~\ref{cla:conind}.
							\end{proof}

							Since,
							%by~$(\mathbf{P0})$,
							by definition,
							we have 
							$0 : \p \in S_{0} \subseteq \mathbf{T}$,
							%there exists $w_{\p} \in \Wmc$ such that $\p \in \qs(w_{\p})$,
							thanks to Claim~\ref{cla:forind} we obtain $\Mmc, 0 \models \p$.
							%for some $w_{\p} \in \Wmc$.
						\end{proof}

						%%% COMPLETENESS
						
						We finally show completeness of the $\LnALC$ tableau algorithm.
						
						\begin{theorem}[Completeness]
							\label{thm:completeness}
							If $\p$ is $\LnALC$ satisfiable, then, having started on the initial completion set $\mathbf{T}_{\p}$, the $\LnALC$ tableau algorithm constructs an $\LnALC$-complete and clash-free completion set for $\p$.
						\end{theorem}
						\begin{proof}
							Let $\Mmc = (\Fmc, \Imc)$ be an $\LnALC$-model satisfying $\p$, with $\Fmc = (\Wmc, \{ \Nmc \}_{i \in I})$, i.e.,
							$\Mmc, w_{\p} \models \p$, for some $w_{\p} \in \Wmc$.
							We require the following definitions and technical results.
							%
							%For every $d \in \Delta_{w}$, define $\tp^{\Imc_{w}}(d) = \{ C \in \conneg(\p) \mid d \in C^{\Imc_{w}} \}$,
							%and let $T_{w} = \{ \tp^{\Imc_{w}}(d) \mid d \in \Delta_{w} \}$.
							%Moreover, for every $t = \tp^{\Imc_{w}}(d)$, select a variable $x_{t} \in \NV$.
							First, we let $\gamma, \delta$ (possibly indexed) range over $\MLnALC$ concepts and formulas, with $\llbracket \gamma \rrbracket^{\Mmc}_{d} = \llbracket \psi \rrbracket^{\Mmc}$, if $\gamma = \psi$, and $\llbracket \gamma \rrbracket^{\Mmc}_{d} = \llbracket C \rrbracket^{\Mmc}_{d}$, if $\gamma = C$.
							Then, for $w \in \Wmc$ and $d \in \bigcup_{v \in \Wmc} \Delta_{v}$, define
							$\Phi^{d}_{w} = \{ \psi \in \forneg(\p) \mid \Mmc, w \models \psi \} \cup \{ C \in \conneg(\p) \mid d \in C^{\Imc_{w}} \}$.
							Observe that, if $C \in \Phi^{d}_{w}$, then $d \in \Delta_{w}$.
							We now show that the following holds.
							%where as usual we let $\gamma, \delta$ range over formulas or concepts.
							%
							\begin{claim}
								\label{cla:truth}
								For every $w \in \Wmc$ and every $d_{1}, \ldots, d_{k}, e \in \bigcup_{v \in \Wmc} \Delta_{v}$:
								if $\Box_{i}\gamma_{1} \in \Phi^{d_{1}}_{w}, \ldots, \Box_{i} \gamma_{k} \in \Phi^{d_{k}}_{w}$ and $\Diamond_{i} \delta \in \Phi^{e}_{w}$, then there exists $v \in \Wmc$ such that:
								\begin{enumerate}[label=$(\arabic*)$, start=0]
									\item $\gamma_{1} \in \Phi^{d_{1}}_{v}, \ldots, \gamma_{k} \in \Phi^{d_{k}}_{v}$ and $\delta \in \Phi^{e}_{v}$; or
									\item $\dnot \gamma_{1} \in \Phi^{d_{1}}_{v}$ and $\dnot \delta \in \Phi^{e}_{v}$; or
									\item[] $\vdots$
									\item[$(l)$] $\dnot \gamma_{l} \in \Phi^{d_{k}}_{v}$ and $\dnot \delta \in \Phi^{e}_{v}$;
								\end{enumerate}
								where:
								for $\mathbf{L} = \mathbf{E}$, $k = l = 1$;
								for $\mathbf{L} = \mathbf{M}$, $k = 1$ and $l = 0$;
								for $\mathbf{L} = \mathbf{C}$, $k \geq 1$ and $l = k$;
								for $\mathbf{L} = \mathbf{N}$, $k = l = 1$ or $k = l = 0$.
							\end{claim}
								\begin{proof}
									
									We consider each $\mathbf{L} \in \mathsf{Log}$.
									
									\begin{enumerate}[leftmargin=*, align=left]
										\item[$\mathbf{L} = \mathbf{E}$.] 
										Assume $\Box_{i} \gamma \in \Phi^{d}_{w}, \Diamond_{i} \delta \in \Phi^{e}_{w}$, meaning that $\llbracket \gamma \rrbracket^{\Mmc}_{d} \in \Nmc_{i}(w)$ and $\Wmc \setminus \llbracket \delta \rrbracket^{\Mmc}_{e} \not \in \Nmc_{i}(w)$, i.e., $\llbracket \dnot \delta \rrbracket^{\Mmc}_{e} \not \in \Nmc_{i}(w)$. Towards a contradiction, suppose that, for every $v \in \Wmc$, the following holds:
										%		\begin{center}
											($\gamma \not \in \Phi^{d}_{v}$ or $\delta \not \in \Phi^{e}_{v}$) and
											($\dnot \gamma \not \in \Phi^{d}_{v}$ or $\dnot \delta \not \in \Phi^{e}_{v}$).
											%		\end{center}
										Equivalently, for every $v \in \Wmc$:
										%		\begin{center}
											($\gamma\in \Phi^{d}_{v}$ implies $\delta \not \in \Phi^{e}_{v}$) and
											($\dnot \delta \in \Phi^{e}_{v}$ implies $\dnot \gamma \not \in \Phi^{d}_{v}$).
											%			($\dnot \gamma \in \Phi^{d}_{v}$ implies $\dnot \delta \not \in \Phi^{e}_{v}$).
											%		\end{center}
										By definition, we have that $\gamma \in \Phi^{d}_{v}$ iff $\dnot \gamma \not \in \Phi^{d}_{v}$ and $\delta \not \in \Phi^{e}_{v}$ iff $\dnot \delta \in \Phi^{e}_{v}$. Thus, the previous step means:
										%		\begin{center}
											($ \llbracket \gamma \rrbracket^{\Mmc}_{d} \subseteq \llbracket \dnot \delta \rrbracket^{\Mmc}_{e}$) and
											($  \llbracket \dnot \delta \rrbracket^{\Mmc}_{e} \subseteq \llbracket \gamma \rrbracket^{\Mmc}_{d}$),
											i.e.,
											%			($\llbracket \dnot \gamma \rrbracket^{\Mmc}_{d} \subseteq \llbracket \delta \rrbracket^{\Mmc}_{e}$).
											%		From this
											%%		since $W \setminus [ \dnot \psi ]^{\Mmc} = \llbracket \psi \rrbracket^{\Mmc}$,
											%		we have equivalently that
											$\llbracket \gamma \rrbracket^{\Mmc}_{d} = \llbracket \dnot \delta \rrbracket^{\Mmc}_{e}$,
											contradicting the assumption that $\llbracket \gamma \rrbracket^{\Mmc}_{d} \in \Nmc_{i}(w)$ and $\llbracket \dnot \delta \rrbracket^{\Mmc}_{e} \not \in \Nmc_{i}(w)$.

											\item[$\mathbf{L} = \mathbf{M}$.] 
											Assume $\Box_{i} \gamma \in \Phi^{d}_{w}, \Diamond_{i} \delta \in \Phi^{e}_{w}$, meaning that $\llbracket \gamma \rrbracket^{\Mmc}_{d} \in \Nmc_{i}(w)$ and $\Wmc \setminus \llbracket \delta \rrbracket^{\Mmc}_{e} \not \in \Nmc_{i}(w)$, i.e., $\llbracket \dnot \delta \rrbracket^{\Mmc}_{e} \not \in \Nmc_{i}(w)$. Towards a contradiction, suppose that, for every $v \in \Wmc$, the following holds:
											$\gamma \not \in \Phi^{d}_{v}$ or $\delta \not \in \Phi^{e}_{v}$.
											Equivalently, for every $v \in \Wmc$:
											$\gamma\in \Phi^{d}_{v}$ implies $\delta \not \in \Phi^{e}_{v}$.
											%	By definition of $\qs(v)$, we have that $\gamma \in \Phi^{d}_{v}$ iff $\dnot \gamma \not \in \Phi^{d}_{v}$ and $\delta \not \in \Phi^{e}_{v}$ iff $\dnot \delta \in \Phi^{e}_{v}$.
											%	Thus,
											By definition,
											the previous step means
											$ \llbracket \gamma \rrbracket^{\Mmc}_{d} \subseteq \llbracket \dnot \delta \rrbracket^{\Mmc}_{e}$.
											Since $\Mmc$ is supplemented, we have that $\llbracket \dnot \delta \rrbracket^{\Mmc}_{e} \in \Nmc_{i}(w)$,
											which is impossible.
											%		contrary to the assumption that $\llbracket \dnot \delta \rrbracket^{\Mmc}_{e} \not \in \Nmc_{i}(w)$.

											\item[$\mathbf{L} = \mathbf{C}$.] 
											Assume $\Box_{i} \gamma_{1} \in \Phi^{d_{1}}_{w}, \ldots, \Box_{i} \gamma_{k} \in \Phi^{d_{k}}_{w}, \Diamond_{i} \delta \in \Phi^{e}_{w}$, meaning that $\llbracket \gamma_{j} \rrbracket^{\Mmc}_{d_{j}} \in \Nmc_{i}(w)$, for $j = 1, \ldots, k$, and $\Wmc \setminus \llbracket \delta \rrbracket^{\Mmc}_{e} \not \in \Nmc_{i}(w)$, i.e., $\llbracket \dnot \delta \rrbracket^{\Mmc}_{e} \not \in \Nmc_{i}(w)$. Towards a contradiction, suppose that, for every $v \in \Wmc$, 
											%it is not the case that the following holds:
											%				\begin{itemize}
												%%				[label=$(\arabic*)$, start=0]
												%					\item $\gamma_{1} \in \Phi^{d_{1}}_{v}, \ldots, \gamma_{k} \in \Phi^{d_{k}}_{v}$ and $\delta \in \Phi^{e}_{v}$; or
												%					\item $\dnot \gamma_{1} \in \Phi^{d_{1}}_{v}$ and $\dnot \delta \in \Phi^{e}_{v}$; or
												%					\item[] $\vdots$
												%					\item $\dnot \gamma_{k} \in \Phi^{d_{k}}_{v}$ and $\dnot \delta \in \Phi^{e}_{v}$.
												%				\end{itemize}
											none of the following holds:
											%it is not the case that the following holds:
											$(0)$ $\gamma_{1} \in \Phi^{d_{1}}_{v}, \ldots, \gamma_{k} \in \Phi^{d_{k}}_{v}$ and $\delta \in \Phi^{e}_{v}$; 
											%or
											$(1)$ $\dnot \gamma_{1} \in \Phi^{d_{1}}_{v}$ and $\dnot \delta \in \Phi^{e}_{v}$; ...;
											%or
											$(k)$ $\dnot \gamma_{k} \in \Phi^{d_{k}}_{v}$ and $\dnot \delta \in \Phi^{e}_{v}$.
											Equivalently, for every $v \in \Wmc$,
											%				\begin{itemize}
												%%				[label=$(\arabic*)$, start=0]
												%					\item $\gamma_{1} \in \Phi^{d_{1}}_{v}, \ldots, \gamma_{k} \in \Phi^{d_{k}}_{v}$ implies $\delta \not \in \Phi^{e}_{v}$; and
												%					\item $\dnot \delta \in \Phi^{e}_{v}$ implies $\dnot \gamma_{1} \not \in \Phi^{d_{1}}_{v}$; and
												%					\item[] $\vdots$
												%					\item $\dnot \delta \in \Phi^{e}_{v}$ implies $\dnot \gamma_{k} \not \in \Phi^{d_{k}}_{v}$.
												%				\end{itemize}
											it holds that
											$(0)$ $\gamma_{1} \in \Phi^{d_{1}}_{v}, \ldots, \gamma_{k} \in \Phi^{d_{k}}_{v}$ implies $\delta \not \in \Phi^{e}_{v}$; and
											$(1)$ $\dnot \delta \in \Phi^{e}_{v}$ implies $\dnot \gamma_{1} \not \in \Phi^{d_{1}}_{v}$; ...
											and 
											$(k)$ $\dnot \delta \in \Phi^{e}_{v}$ implies $\dnot \gamma_{k} \not \in \Phi^{d_{k}}_{v}$.
											By definition, from the previous step we obtain
											%					\begin{itemize}
												%%				[label=$(\arabic*)$, start=0]
												%					\item $\bigcap_{j = 1}^{k} \llbracket \gamma_{j} \rrbracket^{\Mmc}_{d_{j}} \subseteq \llbracket \dnot \delta \rrbracket^{\Mmc}_{e}$; and
												%					\item $\llbracket \dnot \delta \rrbracket^{\Mmc}_{e} \subseteq \llbracket \gamma_{1} \rrbracket^{\Mmc}_{d_{1}}$; and
												%					\item[] $\vdots$
												%					\item $\llbracket \dnot \delta \rrbracket^{\Mmc}_{e} \subseteq \llbracket \gamma_{k} \rrbracket^{\Mmc}_{d_{k}}$.
												%				\end{itemize}
											$(0)$ $\bigcap_{j = 1}^{k} \llbracket \gamma_{j} \rrbracket^{\Mmc}_{d_{j}} \subseteq \llbracket \dnot \delta \rrbracket^{\Mmc}_{e}$; and
											$(1)$ $\llbracket \dnot \delta \rrbracket^{\Mmc}_{e} \subseteq \llbracket \gamma_{1} \rrbracket^{\Mmc}_{d_{1}}$; ...
											and
											$(k)$ $\llbracket \dnot \delta \rrbracket^{\Mmc}_{e} \subseteq \llbracket \gamma_{k} \rrbracket^{\Mmc}_{d_{k}}$.
											Hence $\bigcap_{j = 1}^{k} \llbracket \gamma_{j} \rrbracket^{\Mmc}_{d_{j}} =  \llbracket \dnot \delta \rrbracket^{\Mmc}_{e}$.
											Since $\Mmc$ is closed under intersection, we obtain $\llbracket \dnot \delta \rrbracket^{\Mmc}_{e} \in \Nmc_{i}(w)$,
											a contradiction.

											\item[$\mathbf{L} = \mathbf{N}$.] 
											We distinguish two cases:
											%	\begin{itemize}
												%		\item Let $k = l = 0$. That is, there exists no $\Box_{i} \gamma \in \Phi^{d}_{w}$, while $\Diamond_{i} \delta \in \Phi^{e}_{w}$, meaning that $\Wmc \setminus \llbracket \delta \rrbracket^{\Mmc}_{e} \not \in \Nmc_{i}(w)$.
												%		Towards a contradiction, suppose that, for every $v \in \Wmc$, $\delta \not \in \Phi^{e}_{v}$.
												%			Since, by definition, we have $\delta \not \in \Phi^{e}_{v}$ iff $\dnot \delta \in \Phi^{e}_{v}$, the previous step means that $\Wmc \subseteq \llbracket \dnot \delta \rrbracket^{\Mmc}_{e}$, and hence $\llbracket \delta \rrbracket^{\Mmc}_{e} = \emptyset$. Thus, $\Wmc \not \in \Nmc_{i}(w)$, contradicting the fact that $\Mmc$ contains the unit.
												%
												%		\item Let $k = l = 1$. Hence, there exists $\Box_{i} \gamma \in \Phi^{e}_{w}$ and  $\Diamond_{i} \delta \in \Phi^{e}_{w}$. We then reason similarly to the case for $\mathbf{L} = \mathbf{E}$.
												%	\end{itemize}
											$(i)$ Let $k = l = 0$. That is, there exists no $\Box_{i} \gamma \in \Phi^{d}_{w}$, while $\Diamond_{i} \delta \in \Phi^{e}_{w}$, meaning that $\Wmc \setminus \llbracket \delta \rrbracket^{\Mmc}_{e} \not \in \Nmc_{i}(w)$.
											Towards a contradiction, suppose that, for every $v \in \Wmc$, $\delta \not \in \Phi^{e}_{v}$.
											Since, by definition, we have $\delta \not \in \Phi^{e}_{v}$ iff $\dnot \delta \in \Phi^{e}_{v}$, the previous step means that $\Wmc \subseteq \llbracket \dnot \delta \rrbracket^{\Mmc}_{e}$, and hence $\llbracket \delta \rrbracket^{\Mmc}_{e} = \emptyset$. Thus, $\Wmc \not \in \Nmc_{i}(w)$, contradicting the fact that $\Mmc$ contains the unit.
											$(ii)$ Let $k = l = 1$. Hence, there exists $\Box_{i} \gamma \in \Phi^{e}_{w}$ and  $\Diamond_{i} \delta \in \Phi^{e}_{w}$. We then reason similarly to the case for $\mathbf{L} = \mathbf{E}$.	\qedhere
										\end{enumerate}
									\end{proof}

									Given a completion set $\mathbf{T}$ for $\p$
									and $S_{n} \subseteq \mathbf{T}$,
									%let
									%$\mathsf{L}_{\mathbf{T}} = \{ n \in \mathsf{N_{L}} \mid S_{n} \subseteq \mathbf{T} \}$.
									%Moreover,
									let $\Gamma^{x}_{n} = \{ \psi \mid n : \psi \in S_{n} \} \cup \{ C \mid n : C(x) \in S_{n} \}$.
									We say that a completion set $\mathbf{T}$ for $\p$ is \emph{$\Mmc$-compatible} if
									there exists a function $\pi$ from $\mathsf{L}_{\mathbf{T}}$ to $\Wmc$, and, for every $n \in \mathsf{L}_{\mathbf{T}}$, there exists a function $\pi_{n}$ from the set of variables occurring in $S_{n}$ to $\Delta_{\pi(n)}$, such that
									$\gamma \in \Gamma^{x}_{n}$ implies $\gamma \in \Phi^{\pi_{n}(x)}_{\pi(n)}$.
									%\pi(n) \in \llbracket \gamma \rrbracket^{\Mmc}_{\pi_{n}(x)}$.
									%\nb{M: todo fix}
									%\begin{itemize}
									%	\item there exists a function $\pi \colon \mathsf{L}_{\mathbf{T}} \to \Wmc$
									%%	from $N$ to $\Wmc$
									%%	the set of labels of the labelled constraints in $\mathbf{T}$
									%	such that $n : \psi \in S_{n}$ implies $\Mmc, \pi(n) \models \psi$, for every $\psi \in \for(\p)$;
									%	\item for every $n \in \mathsf{L}_{\mathbf{T}}$, there exists a function $\pi_{n}$ from the set of variables occurring in $S_{n}$ to $\Delta_{\pi(n)}$ such that $n: C(x) \in S_{n}$ implies $\pi_{n}(x) \in C^{\Imc_{\pi(n)}}$.
									%\end{itemize}
									We then require the following claim.
									
									\begin{claim}
										\label{cla:compatible}
										If a completion set $\mathbf{T}$ for $\p$ is $\Mmc$-compatible, then for every  $\LnALC$-rule $\mathsf{R}$ applicable to $\mathbf{T}$ there exists a completion set $\mathbf{T}'$ obtained from $\mathbf{T}$ by an application of $\mathsf{R}$ such that $\mathbf{T}'$ is $\Mmc$-compatible.
										%If a completion set $\mathbf{T}$ for $\p$ is $\Mmc$-compatible and $\mathbf{T}'$ is obtained from $\mathbf{T}$ by an application of an $\LnALC$-rule $\mathsf{R}$, then $\mathbf{T}'$ is $\Mmc$-compatible.
									\end{claim}
									\begin{proof}
										Given an $\Mmc$-compatible completion set $\mathbf{T}$ for $\p$ and a label $n \in \mathsf{L}_{\mathbf{T}}$, let $\pi$ and $\pi_{n}$ be the functions provided by the definition of $\Mmc$-compatibility.
										We need to consider each $\LnALC$-rule $\mathsf{R}$.
										For $\mathsf{R} \in \{ \mathsf{R}_{\land}, \mathsf{R}_{\lor}, \mathsf{R}_{\sqcap}, \mathsf{R}_{\sqcup}, \mathsf{R}_{\forall}, \mathsf{R}_{\exists}, \mathsf{R}_{=}, \mathsf{R}_{\neq} \}$, we proceed similarly to~\cite[Claim 15.14]{GabEtAl03}.
										Here we consider the case of $\mathsf{R}_{\mathbf{L}}$:
										Suppose that $\mathsf{R}_{\mathbf{L}}$ is applicable to $\mathbf{T}$.
										Let $\Box_{i} \gamma_{1} \in \Gamma^{x_{1}}_{n}, \ldots, \Box_{i} \gamma_{k} \in \Gamma^{x_{k}}_{n}, \Diamond_{i} \delta \in \Gamma^{y}_{n}$.
										Since $\mathbf{T}$ is $\Mmc$-compatible,
										we have that $\Box_{i}\gamma_{1} \in \Phi^{\pi_{n}(x_{1})}_{\pi(n)}, \ldots, \Box_{i} \gamma_{k} \in \Phi^{\pi_{n}(x_{k})}_{\pi(n)}$ and $\Diamond_{i} \delta \in \Phi^{\pi_{n}(y)}_{\pi(n)}$.
										Thus, by Claim~\ref{cla:truth}, there exists $v \in \Wmc$ such that:
										%\begin{itemize}
										%\item $\gamma_{1} \in \Phi^{\pi_{n}(x_{1})}_{v}, \ldots, \gamma_{k} \in \Phi^{\pi_{n}(x_{k})}_{v}$ and $\delta \in \Phi^{\pi_{n}(y)}_{v}$; or
										%\item $\dnot \gamma_{j} \in \Phi^{\pi_{n}(x_{j})}_{v}$ and $\dnot \delta \in \Phi^{\pi_{n}(y)}_{v}$,
										%for some $j\leq l$;
										%\end{itemize}
										$\gamma_{1} \in \Phi^{\pi_{n}(x_{1})}_{v}, \ldots, \gamma_{k} \in \Phi^{\pi_{n}(x_{k})}_{v}$ and $\delta \in \Phi^{\pi_{n}(y)}_{v}$; or
										$\dnot \gamma_{j} \in \Phi^{\pi_{n}(x_{j})}_{v}$ and $\dnot \delta \in \Phi^{\pi_{n}(y)}_{v}$,
										for some $j\leq l$;
										%$(0)$ $\gamma_{1} \in \Phi^{\pi_{n}(x_{1})}_{v}, \ldots, \gamma_{k} \in \Phi^{\pi_{n}(x_{k})}_{v}$ and $\delta \in \Phi^{\pi_{n}(y)}_{v}$; or
										%$(1)$ $\dnot \gamma_{1} \in \Phi^{\pi_{n}(x_{1})}_{v}$ and $\dnot \delta \in \Phi^{\pi_{n}(y)}_{v}$; or
										%...
										%$(l)$ $\dnot \gamma_{k} \in \Phi^{\pi_{n}(x_{k})}_{v}$ and $\dnot \delta \in \Phi^{\pi_{n}(y)}_{v}$;
										where:
										for $\mathbf{L} = \mathbf{E}$, $k = l = 1$;
										for $\mathbf{L} = \mathbf{M}$, $k = 1$ and $l = 0$;
										for $\mathbf{L} = \mathbf{C}$, $k \geq 1$ and $l = k$;
										for $\mathbf{L} = \mathbf{N}$, $k = l = 1$ or $k = l = 0$.
										By applying the rule $\mathsf{R}_{\mathbf{L}}$ accordingly, one can obtain $\mathbf{T}'$ by adding
										$m : \gamma_1, \ldots, m : \gamma_k, m : \delta$, or $m : \dot{\lnot}\gamma_j, m : \dnot \delta $, for some $j\leq l$, to $\mathbf{T}$
										(recall that $m$ is fresh for $\mathbf{T}$ and $\gamma_{j}$ is either $\psi_{j} \in \forneg(\p)$ or $C_{j}(x_{j})$, with $C_{j} \in \conneg(\p)$, for $j = 1, \ldots, k$, and $\delta$ is either $\chi \in \forneg(\p)$ or $D(y)$, with $D \in \conneg(\p)$).
										%	 The application of $\mathsf{R}_{\mathbf{L}}$ non-deterministically chooses to add to $\mathbf{T}$ either $\{ m : \gamma_1, \ldots, m: \gamma_k, m: \delta\}$, or $\{ m: \dot{\lnot}\gamma_j, o: \dot{\lnot}\delta\}$, for some $j\leq l$, where $m$ is the $\ll$-minimal label fresh for $\mathbf{T}$.
										%	 We set $\pi(m) = \ldots$
										By extending $\pi$ with $\pi(m) = v$, and $\pi_{m}$ with $\pi_{m}(x_{1}) = \pi_{n}(x_{1})$, \ldots, $\pi_{m}(x_{k}) = \pi_{n}(x_{k})$, $\pi_{m}(y) = \pi_{n}(y)$, we obtain that $\mathbf{T}'$ is $\Mmc$-compatible.
									\end{proof}
									
									To conclude, let $\mathbf{T}_{\p} = \{0 : \p, 0 : \top(x) \}$ be the initial completion set for $\p$.
									%with $0$ being the label in $\mathbf{T}_{\p}$ and $x$ be the variable occurring in $\mathbf{T}_{\p}$.
									Define $\pi(0) = w_{\p}$ (where $\Mmc, w_{\p} \models \p$) and $\pi_{0}(x) = d$, for an arbitrary $d \in \Delta_{w_{\p}}$.
									Clearly, these functions ensure that $\mathbf{T}_{\p}$ is $\Mmc$-compatible.
									By Claim~\ref{cla:compatible}, we can apply the $\LnALC$-rules so that the obtained completion sets are $\Mmc$-compatible as well.
									From Theorem~\ref{thm:termination}, we have that the $\LnALC$ tableau algorithm eventually terminates, returning an $\LnALC$-complete completion set for $\p$ that is clash-free by construction.
								\end{proof}

								By Theorem~\ref{thm:termination}, we have that the non-deterministic $\LnALC$ tableau algorithm terminates after exponentially many steps in the size of the input formula. By Theorems~\ref{thm:soundness} and~\ref{thm:completeness}, such algorithm is sound and complete with respect to satisfiability in varying domain neighbourhood models. Thus, we obtain the following result.
								%\nb{M: todo fix notation in Preliminaries}
								
								\begin{theorem}
									\label{thm:upperbound}
									The $\LnALC$ formula satisfiability problem on varying domain neighbourhood models is decidable in $\NExpTime$.
								\end{theorem}

								To conclude this section, we observe that as an immediate consequence of the above results we also obtain 
								a (constructive) proof of the following kind of \emph{exponential model property}.
								%\begin{corollary}
								%For each $\mathbf{L} \in \mathsf{Log}$,
								%every $\LnALC$-satisfiable formula $\p$ has a model of at most exponential size in the length of $\p$.
								%\end{corollary}
								%\begin{corollary}
								%For each $\mathbf{L} \in \{\mathbf{E},\mathbf{M},\mathbf{N}\}$,
								%every $\LnALC$-satisfiable formula $\p$ has a model 
								%%of at most exponential size in the length of $\p$.
								%with at most $p(|\fg(\p)|)$ worlds, each of them with a domain of at most 
								%$2^{q(|\fg(\p)|)}$ elements, where $p$ and $q$ are polynomial functions.
								%Moreover, for $\mathbf{L} = \mathbf{C}$,
								%every $\LnALC$-satisfiable formula $\p$ has a model 
								%with at most $2^{p(|\fg(\p)|)}$ worlds, each of them with a domain of at most 
								%$2^{q(|\fg(\p)|)}$ elements, where $p$ and $q$ are polynomial functions.
								%\end{corollary}
								\begin{corollary}
									For $\mathbf{L} \in \{\mathbf{E},\mathbf{M},\mathbf{N}\}$
									(respectively, $\mathbf{L} = \mathbf{C}$),
									every $\LnALC$ satisfiable formula $\p$ has a model 
									%of at most exponential size in the length of $\p$.
									with at most $p(|\fg(\p)|)$ (respectively, at most $2^{p(|\fg(\p)|)}$) worlds, each of them having a domain with at most 
									$2^{q(|\fg(\p)|)}$ elements, where $p$ and $q$ are polynomial functions.
								\end{corollary}
								\begin{proof}
									By Theorem \ref{thm:completeness}, if $\p$ is $\LnALC$ satisfiable, then 
									there is a $\LnALC$-complete and clash-free completion set $\mathbf{T}$ for it.
									Then by Theorem \ref{thm:soundness},
									%basing on $\mathbf{T}$ we can define a
									there exists a model 
									$\Mmc = (\Wmc, \{ \Nmc_{i} \}_{i \in I}, \Imc)$
									for $\p$ where $\Wmc =  \mathsf{L}_{\mathbf{T}}$
									and for each $n\in\Wmc$, $\Delta_{n} = \{ x \in \mathsf{N_{V}} \mid x \ \text{occurs in} \ S_{n} \}$.
									%By Theorem~\ref{thm:termination}, Claim~\ref{cla:termglobal},
									%$|\mathsf{L}_{\mathbf{T}}| \leq |\fg(\p)|^2$ for $\mathbf{L} \in \{\mathbf{E}, \mathbf{M}, \mathbf{N}\}$,
									%and 
									%$|\mathsf{L}_{\mathbf{T}}| \leq 2^{|\fg(\p)|} \cdot |\fg(\p)|$ for $\mathbf{L} = \mathbf{C}$,
									%finally by Theorem~\ref{thm:termination}, Claim~\ref{cla:termlocal}, 
									%$|\{ x \in \mathsf{N_{V}} \mid x \ \text{occurs in} \ S_{n} \}|$ does not exceed $2^{q(|\fg(\p)|)}$,
									%where $p$ and $q$ are polynomial functions.
									By Theorem~\ref{thm:termination}, Claim~\ref{cla:termglobal}, it follows
									$|\Wmc| \leq |\fg(\p)|^2$ for $\mathbf{L} \in \{\mathbf{E}, \mathbf{M}, \mathbf{N}\}$,
									and 
									$|\Wmc| \leq 2^{|\fg(\p)|} \cdot |\fg(\p)|$ for $\mathbf{L} = \mathbf{C}$,
									finally by Theorem~\ref{thm:termination}, Claim~\ref{cla:termlocal}, 
									for each $n\in\Wmc$, $|\Delta_n|$ does not exceed $2^{q(|\fg(\p)|)}$,
									where $p$ and $q$ are polynomial functions.
								\end{proof}

%% file: sec_frag.tex
\section{Reasoning in Fragments without Modalised Concepts}

\newcommand{\MLn}{\ensuremath{\smash{\mathit{ML}^{n}}}\xspace}
\newcommand{\Nn}{\ensuremath{\smash{\mathbf{N}^{n}}}\xspace}

An \emph{$\MLnALCg$ formula} is defined similarly to the $\MLnALC$ case,
%but inductively built from $\ALC$ \emph{atoms} consisting of CIs (cf.~Section~\ref{sec:prelim}), or \emph{assertions} of the form $A(a)$ and $r(a,b)$, where $A \in \NC$, $r \in \NR$, and $a,b \in \NI$,
%and
by disallowing modalised concepts.
Given
$\mathbf{L} \in \mathsf{Log}$,
%$\mathbf{L} \in \{ \mathbf{E}, \mathbf{M}, \mathbf{C}, \mathbf{N} \}$,
the \emph{$\LnALCg$ formula satisfiability problem on constant domain neighbourhood models} is the $\MLnALCg$ formula satisfiability problem on constant domain neighbourhood models based on neighbourhood frames in the respective class for $\mathbf{L}$ (cf. Section~\ref{sec:prelim}).
%
%(i.e., built from $\ALC$ concepts and CIs only).
An \emph{$\MLn$ formula}, instead, is defined analogously to $\MLnALCg$, except that we built it from the standard propositional (rather than $\ALC$) language over a countable set of \emph{propositional letters} $\mathsf{N_{P}}$.
%consider only propositional logic (instead of \ALC).
The semantics of
$\MLn$ formulas
is given in terms of \emph{propositional 
neighbourhood models} (or simply \emph{models}) $\Mmc^{\sf P} = (\Wmc, \{ \Nmc_{i} \}_{i \in I}, \Vmc)$,
where $(\Wmc, \{ \Nmc_{i} \}_{i \in I})$ is a neighbourhood frame,
with $I = \{ 1, \ldots, n \}$ in the following,
and $\Vmc: \NPr \rightarrow 2^{\Wmc}$ is a function 
mapping propositional letters
%in $\NPr$
to
sets of worlds
%subsets of the domain of worlds
(see~\cite{Che,Var2}).
%We say that a propositional 
%neighbourhood model is a $\C^{n}$ or $\Nn$ \emph{model} if it is based on a neighbourhood frame satisfying the corresponding conditions for $\C^{n}$ and $\Nn$ given in Section~\ref{sec:prelim}, respectively.
The \emph{$\ensuremath{\smash{\mathbf{L}^{n}}}$ formula satisfiability problem}, is the $\MLn$ formula satisfiability problem on propositional neighbourhood models based on neighbourhood frames in the respective class for $\mathbf{L}$.
A propositional neighbourhood model based on a neighbourhood frame in the respective class for $\mathbf{L}$ is called \emph{$\mathbf{L}^{n}$ model}.
%, supplemented N-frames, {N-frames closed under intersection, and N-frames containing the unit, respectively. 
%\nb{double check if
%considering all $\mathbf{L} \in \mathsf{Log}$ is ok}

%%%%%%%%%%%%%%%
%global gcis

%\paragraph{Satisfiability in $\CALCg$  and $\NALCg$}
%\paragraph{{\bf Satisfiability in $\CALCg$  and $\NALCg$}}
%\nb{O: working here}
In~\citeauthor{DL19}~\cite{DL19}, it is shown that 
$\EALCg$  and $\MALCg$ formula satisfiability problems on constant domain neighbourhood models are 
$\ExpTime$-complete.
We now show tight complexity results %$\ExpTime$ upper bounds 
for $\CALCg$  and $\NALCg$,
using again the notion of 
a propositional abstraction of a formula 
(as in, e.g.,~\cite{Baader:2012:LOD:2287718.2287721}).
%Since $\ALC$ formula satisfiability is already $\ExpTime$-hard, 
%our upper bounds here are tight. 
%we have 
%a tight complexity result for the global cases.
%We show that
Here, one can separate the satisfiability test into two parts, 
one for the description logic dimension and 
one for the 
%dimension of the
%\neighborhood
%frame. 
modal dimension.
%~ In this subsection, we also consider the lightweight DL called \EL, defined 
%~ as the fragment of \ALC which only allows conjunctions and existential quantification 
%~ in concept expressions. 
%
%Consider $\Lmc\in\{\ALC,\EL\}$. 
%\newcommand{\propmodel}{\ensuremath{M}\xspace}
%\newcommand{\propdomain}{\ensuremath{W}\xspace}
%\newcommand{\propneigh}{\ensuremath{N}\xspace}
%\newcommand{\propassign}{\ensuremath{I}\xspace}
The 
\emph{propositional  abstraction} $\prop{\varphi}$ of an
$\MLnALCg$
%$\CALCg$ (respectively, $\NALCg$)
formula $\varphi$ is  
the result of replacing each $\ALC$
CI
%atom
in $\varphi$ by 
a propositional letter $p$, so that there is a $1:1$ relationship 
between the $\ALC$
CI
%atoms
$\elaxiom$ occurring in $\varphi$ and the 
propositional letters $p_{\elaxiom}$ used for the abstraction.  
%The semantics of $\prop{\varphi}$ is given in terms of \emph{propositional 
%neighbourhood models} $ (\Wmc, \{ \Nmc_{i} \}_{i \in I}, \Vmc)$ for $\C^{n}$ (respectively, $\Nn$), 
%where $(\Wmc, \{ \Nmc_{i} \}_{i \in I})$ is a neighbourhood frame 
%and $\Vmc: \NPr \rightarrow 2^{\Wmc}$ is a function 
%mapping propositional letters in $\NPr$ to
%sets of worlds
%%subsets of the domain of worlds
%(see~\cite{Che,Var2}). 
We set $\NPr(\varphi) = \{p_{\elaxiom}\in\NPr\mid \elaxiom \text{ is an \ALC
CI
%atom
in }
\varphi \}$.
Given an
$\MLnALCg$
%\nb{M: todo fix}
%$\CALCg$ (respectively, $\NALCg$)
formula $\varphi$, we say that a propositional neighbourhood model 
$\propmodel = (\Wmc, \{ \Nmc_{i} \}_{i \in I}, \Vmc)$
of $\prop{\varphi}$
is \emph{$\consistent{\varphi}$}
if, for all $w\in \Wmc$,
the following \ALC formula is satisfiable $$\textstyle\bigwedge_{p_{\elaxiom}\in \NPr(w)} \ {\elaxiom} \ \wedge \
\bigwedge_{p_{\elaxiom}\in \overline{\NPr(w)}}\ \neg {\elaxiom},$$
where $\NPr(w) = \{p_{\elaxiom}\in \NPr(\varphi) \mid w\in \Vmc(p_{\elaxiom})\}$
and $\overline{\NPr(w)}=\NPr(\varphi)\setminus\NPr(w)$. 
We now formalise the connection between %complexity of the satisfiability problem for 
$\MLnALCg$
%$\CALCg$ and $\NALCg$
formulas and their propositional abstractions %with consistent models
with the following lemma, where $\mathbf{L} \in \{ \mathbf{C}, \mathbf{N} \}$, obtained by adapting the proof of~\citeauthor{DL19}~\cite[Lemma~1]{DL19}.
%, which is an adaptation of the 
%results obtained for other \ALC extensions~\cite{Baader:2012:LOD:2287718.2287721}.
%~ %We say that a model of $\prop{\varphi}$
%~ %is $\varphi$-consistent

%\nb{$\varphi$-consistent model}

\begin{restatable}{lemma}{Lemmaprop}\label{lem:prop}
A formula $\varphi$ is $\LnALCg$ satisfiable
on constant domain neighbourhood models
iff
%if, and only if, 
$\prop{\varphi}$ is satisfied in a $\varphi$-consistent $\mathbf{L}^{n}$ model.  
\end{restatable}

%\begin{restatable}{lemma}{Lemmaprop}\label{lem:general}
%A formula $\prop{\varphi}$ is satisfied in a $\varphi$-consistent  
%$\mathbf{L}^{n}$
%model
%iff
%if, and only if,
%there is   a $\varphi$-consistent valuation \valuation 
%for $\prop{\varphi}$ such that 
%if $\B_i\psi_1,   \B_i\psi_2$ are in 
%${\sf sub}(\prop{\varphi})$, $\valuation(\B_i\psi_1)=1$   and 
%$\valuation(\B_i\psi_2)=0$, then either $(\psi_1\wedge\neg\psi_2)$ or 
%$(\neg\psi_1\wedge\psi_2)$ 
%is satisfied in a $\varphi$-consistent $\mathbf{L}^{n}$
 %model. 
%\begin{itemize}
%\item if $\B_i\psi$ is in ${\sf sub}(\prop{\varphi})$ and
%$\valuation(\B_i\psi_1)=0$ then $\neg \psi$ is 
%$\C^{n}$ satisfiable; and 
%\item 
%\end{itemize}
%\end{restatable}

We
%\nb{M: moved footnote here to save space}
%\nb{M: todo fix notation}
assume that the 
primitive
%propositional
connectives used to build 
%the
propositional formulas are $\neg$ %, $\vee$, 
and $\wedge$ ($\vee$ is expressed using 
$\neg$ %, $\vee$, 
and $\wedge$),
%Moreover,
and
we
denote by  ${\sf sub}(\prop{\varphi})$ the set 
of subformulas of $\prop{\varphi}$ closed under single negation.  
A \emph{valuation} for a propositional
$\MLn$
%modal logic
formula $\prop{\varphi}$   
is a function $\nu: {\sf sub} (\prop{\varphi})\rightarrow \{0,1\}$ that 
satisfies the following conditions:
%\footnote{Assuming that the 
%primitive propositional connectives used to build 
%the formulas are $\neg$ %, $\vee$, 
%and $\wedge$ ($\vee$ is expressed using 
%$\neg$ %, $\vee$, 
%and $\wedge$).}:
(1) for all $\neg\psi\in {\sf sub} (\prop{\varphi})$,
$\nu(\psi)=1$ iff $\nu(\neg\psi)=0$;
%(2) for all $\psi_1\vee \psi_2\in {\sf sub} (\prop{\varphi})$,
%$\nu(\psi_1\vee \psi_2) = 1$ iff $\nu(\psi_1) = 1$
%or $\nu(\psi_2) = 1$; 
(2) for all $\psi_1\wedge \psi_2\in {\sf sub} (\prop{\varphi})$,
$\nu(\psi_1\wedge \psi_2) = 1$ iff $\nu(\psi_1) = 1$
and $\nu(\psi_2) = 1$; 
and (3) $\nu(\prop{\varphi})=1$. 
We say that a valuation for $\prop{\varphi}$ 
is \emph{$\varphi$-consistent} if any
propositional
neighbourhood
model of the form $(\{w\}, \{ \propneigh_{i} \}_{i \in I}, \propassign)$ satisfying
$w\in \propassign(p_\elaxiom)$ iff $\nu(p_\elaxiom)=1$, for all $p_\elaxiom\in \NPr(\varphi)$, 
 is $\varphi$-consistent.
We now establish that satisfiability of 
$\prop{\varphi}$ in a $\varphi$-consistent $\C^{n}$ (respectively, $\Nn$) model is characterized 
by the existence of a  $\varphi$-consistent valuation 
satisfying the property described in Lemma~\ref{lem:proplemma2} 
(respectively, Lemma~\ref{lem:proplemma}).

\begin{restatable}{lemma}{Lemmapropsec}\label{lem:proplemma2}
A formula $\prop{\varphi}$ is satisfied in a $\varphi$-consistent $\C^{n}$ %\nb{check macro}
model
iff
%if, and only if,
there is   a $\varphi$-consistent valuation \valuation 
for $\prop{\varphi}$ such that 
if $\B_i\psi_1, \dots, \B_i\psi_k$ are in 
${\sf sub}(\prop{\varphi})$, $\valuation(\B_i\psi_j)=1$ for all $1\leq j < k$, and 
$\valuation(\B_i\psi_k)=0$, then either $(\bigwedge^{k-1}_{j=1}\psi_j\wedge\neg\psi_k)$ or 
$(\neg\psi_j\wedge\psi_k)$ for some $1\leq j < k$
is satisfied in a $\varphi$-consistent $\C^{n}$ %\nb{check macro}
 model. 
%\begin{itemize}
%\item if $\B_i\psi$ is in ${\sf sub}(\prop{\varphi})$ and
%$\valuation(\B_i\psi_1)=0$ then $\neg \psi$ is 
%$\C^{n}$ satisfiable; and 
%\item 
%\end{itemize}
\end{restatable}
\begin{proof}
%\paragraph{Point 2}
($\Rightarrow$) Suppose that $\prop{\varphi}$ is satisfied in a world $w$ of a $\varphi$-consistent $\C^{n}$ model 
 $\propmodel = (\propdomain, \{ \propneigh_{i} \}_{i \in I}, \propassign)$. That is, 
 $\propmodel, w\models \prop{\varphi}$. We define a $\varphi$-consistent valuation for 
 $\prop{\varphi}$
 by setting $\nu(\psi)=1$ if $\propmodel, w\models \psi$ and $\nu(\psi) = 0$
 if  $\propmodel, w\not\models \psi$. 
 It is easy to check that $\nu$ is indeed a 
 $\varphi$-consistent valuation (given that $\propmodel$ is a  
 $\varphi$-consistent $\C^{n}$ model). Assume  $\B_i\psi_1, \dots, \B_i\psi_k$ are in 
${\sf sub}(\prop{\varphi})$, $\valuation(\B_i\psi_j)=1$ for all $1\leq j < k$, and 
$\valuation(\B_i\psi_k)=0$.
Then $\propmodel,w\models \B_i\psi_j$ for all $1\leq j < k$,
 and $\propmodel,w\not\models \B_i\psi_k$.
By definition, $(\B_i \psi_1 \wedge \ldots \wedge \B_i \psi_{k-1})\rightarrow \B_i (\psi_1\wedge \ldots \wedge \psi_{k-1})$ holds in $\C^{n}$ models.
So $\propmodel,w\models \B_i (\psi_1\wedge \ldots \wedge \psi_{k-1})$ 
 and $\propmodel,w\not\models \B_i\psi_k$.
This means that $\valuation(\B_i(\bigwedge^{k-1}_{j=1}\psi_j))=1$
while $\valuation(\B_i\psi_k)=0$.
  By definition, 
 $\propassign(\bigwedge^{k-1}_{j=1}\psi_j)\in \propneigh_i(w)$   and 
  $\propassign(\psi_k)\not\in \propneigh_i(w)$.
%Using the same argument of Lemma 3.1 in~\cite{Var2}, it follows that 
So,
   $\propassign(\bigwedge^{k-1}_{j=1}\psi_j)\neq \propassign(\psi_k)$.
Then, %it is easy to see that 
there 
   is a world $u$ in the symmetrical difference of these sets 
   such that $\propmodel,u\models (\bigwedge^{k-1}_{j=1}\psi_j\wedge\neg\psi_k)\vee(\neg(\bigwedge^{k-1}_{j=1}\psi_j)\wedge\psi_k)$. 
   
%Assume  $\B_i\psi_1, \dots, \B_i\psi_k$ are in 
%${\sf sub}(\prop{\varphi})$, $\valuation(\B_i\psi_j)=1$ for all $1\leq j < k$, and 
%$\valuation(\B_i\psi_k)=0$.
%By definition, $(\B_i \psi_1 \wedge \ldots \wedge \B_i \psi_{k-1})\rightarrow \B_i %(\psi_1\wedge \ldots \wedge \psi_{k-1})$ holds in $\C^{n}$ models.
%\nb{add this def somewhere} 
%This means that $\valuation(\B_i(\bigwedge^{k-1}_{j=1}\psi_j))=1$.

($\Leftarrow$) Suppose there is a $\varphi$-consistent valuation $\nu$ for $\prop{\varphi}$ such that 
%there is   a $\varphi$-consistent valuation \valuation 
%for $\prop{\varphi}$ such that 
if $\B_i\psi_1, \dots, \B_i\psi_k$ are in 
${\sf sub}(\prop{\varphi})$, $\valuation(\B_i\psi_j)=1$ for all $1\leq j < k$, and 
$\valuation(\B_i\psi_k)=0$, then 
there is a $\varphi$-consistent $\C^{n}$ model $$\propmodel_{\bigwedge^{k-1}_{j=1}\psi_j,\psi_k}=(\propdomain_{\bigwedge^{k-1}_{j=1}\psi_j,\psi_k},\{ \propneigh_{{\bigwedge^{k-1}_{j=1}\psi_j,\psi_k}_{i}} \}_{i \in I},\propassign_{\bigwedge^{k-1}_{j=1}\psi_j,\psi_k})$$ and a world 
$w_{\bigwedge^{k-1}_{j=1}\psi_j,\psi_k}\in \propdomain_{\bigwedge^{k-1}_{j=1}\psi_j,\psi_k}$ such that 
$$\propmodel_{\bigwedge^{k-1}_{j=1}\psi_j,\psi_k}, w_{\bigwedge^{k-1}_{j=1}\psi_j,\psi_k}\models ((\bigwedge^{k-1}_{j=1}\psi_j)\wedge\neg\psi_k)\vee(\neg(\bigwedge^{k-1}_{j=1}\psi_j)\wedge\psi_k).$$ 

%either $(\bigwedge^{k-1}_{j=1}\psi_j\wedge\neg\psi_k)$ or 
%$(\neg\psi_j\wedge\psi_k)$ for some $1\leq j < k$
%is satisfied in a $\varphi$-consistent $\C^{n}$ %\nb{check macro}
% model. 
% if $\B_i\psi_1$ and $\B_i\psi_2$ are in 
%${\sf sub}(\prop{\varphi})$, $\valuation(\B_i\psi_1)=1$, and 
%$\valuation(\B_i\psi_2)=0$, then 
%there is a model $\propmodel_{\psi_1,\psi_2}=(\propdomain_{\psi_1,\psi_2},
%\{ \propneigh_{{\psi_1,\psi_2}_{i}} \}_{i \in I},\propassign_{\psi_1,\psi_2})$
% and a world 
%$w_{\psi_1,\psi_2}\in \propdomain_{\psi_1,\psi_2}$ such that 
%$\propmodel_{\psi_1,\psi_2}, w_{\psi_1,\psi_2}\models (\psi_1\wedge\neg\psi_2)\vee(\neg\psi_1\wedge\psi_2)$. 
Let $\propmodel_1,\ldots,\propmodel_m$ be an enumeration of the models 
$\propmodel_{\bigwedge^{k-1}_{j=1}\psi_j,\psi_k}$, as above.  That is, we take one model $\propmodel_{\bigwedge^{k-1}_{l=1}\psi_l,\psi_k}$ for each 
pair $j=\bigwedge^{k-1}_{l=1}\psi_l,\psi_k$ 
where $\propmodel_j = (\propdomain_j, \{ \propneigh_{j_{i}} \}_{i \in I},\propassign_j)$, 
and let $w_1,\ldots,w_m$ be an enumeration of the worlds
 $w_{\bigwedge^{k-1}_{l=1}\psi_l,\psi_k}$, 
with $j=\bigwedge^{k-1}_{l=1}\psi_l,\psi_k$ and  $w_j\in \propdomain_j$. We assume without loss of generality that $\propdomain_j\cap \propdomain_k=\emptyset$ 
for $j\neq k$. 

In the following, we define a $\varphi$-consistent $\C^{n}$ model   $\propmodel = (\propdomain,\{ \propneigh_{i} \}_{i \in I}, \propassign)$ 
for $\prop{\varphi}$. 
Intuitively, we construct $\propmodel$ by taking the union of each 
$\propmodel_j$ with the addition of a new world $w$ that 
will satisfy $\prop{\varphi}$. 
We define $\propdomain$ as $\bigcup_{1\leq j\leq n}\propdomain_j\cup \{w\}$, 
where $w$ is fresh.
%The tricky part of the proof is to define the assignment $\propassign$. 
Before defining $\propneigh_{i}$ and $\propassign$, we define the function $J: {\sf sub}(\prop{\varphi})\rightarrow 2^{\Wmc}$
with $J(\psi)=\bigcup_{0\leq j \leq m} \Vmc_j(\psi)$ for all $\psi\in {\sf sub}(\prop{\varphi})$, where %$I_i$ is as above for $1\leq i\leq n$, 
 %and
$\Vmc_0: {\sf sub}(\varphi)\rightarrow  2^{\{w\}}$ is the function
that assigns $\psi$ to $\{w\}$, if $\nu(\psi)=1$, 
and to $\emptyset$, otherwise ($\Vmc_j$, for $1\leq j\leq m$, is as above).
By construction, we have that $J(\neg \psi)=\propdomain\setminus J(\psi)$
and $J(\psi_1\wedge \psi_2)=J(\psi_1)\cap J(\psi_2)$. 
We define the assignment $\propassign$ as the function 
$\propassign: \NPr(\varphi)\rightarrow 2^{\Wmc}$ satisfying 
 $\propassign(p_\elaxiom)=J(p_\elaxiom)$ for all $p_\elaxiom\in \NPr(\varphi)$. 
 
It remains to define $\propneigh_i$, for
$i \in I$.
%$1 \leq i \leq n$. 
For $u\in \propdomain_j$ we first put $\alpha \subseteq W$ in $\propneigh_i(u)$ 
precisely when 
%\begin{itemize}
%\item 
$\propmodel_j,u \models \B_i \psi_\alpha$ and $\alpha = J(\psi_\alpha)$
for some $\B_i\psi_\alpha \in {\sf sub}(\varphi)$.
Then, we close $\propneigh_i$ under intersection so that $\propmodel$ is a $\C^{n}$ model. The next two claims establish that $\propneigh_i$ is as expected.
%; and
%\item 
%\end{itemize}
%We claim that 
\begin{claim}
If $\beta \in \Nmc_i(u)$ and $\beta = J(\psi)$ for some $\B_i\psi\in{\sf sub}(\prop{\varphi})$,
then $\propmodel_j,u\models\B_i\psi$.
\end{claim}
Indeed, since $\beta = J(\psi)\in \Nmc_i(u)$, 
we must have that $\propmodel_j,u\models \B_i\psi_{1,\beta}$, \ldots, 
$\propmodel_j,u\models \B_i\psi_{m,\beta}$ and  $\beta = \bigcap^{m}_{l=1} J(\psi_{l,\beta})$ for 
some $\B_i\psi_{1,\beta}, \ldots, \B_i\psi_{m,\beta} \in{\sf sub}(\prop{\varphi})$. %
Since $\propneigh_i$ is closed under intersection,
in fact, we have that $\propmodel_j,u \models \B_i (\bigwedge^m_{l=1}\psi_{l,\beta})$.
%\nb{intersection}
But since
$J(\psi)=\bigcap^{m}_{i=1} J(\psi_{i,\beta})$, we also have $\propassign_j(\psi)=\bigcap^{m}_{l=1}\propassign_j(\psi_{l,\beta})$ 
(recall that $\propdomain_j \cap \propdomain_k = \emptyset$ for  
$k \neq j$), so $\propmodel_j,u\models \B_i\psi$ iff 
$\propmodel_j,u \models \B_i (\bigwedge^m_{l=1}\psi_{l,\beta})$.
It follows that $\propmodel_j,u\models \B_i\psi$.

\medskip

Regarding the fresh 
world $w$ introduced above in $\propdomain$,
we first put $\alpha \subseteq \Wmc$ in $\propneigh_i(w)$   precisely 
when $\nu(\B_i\psi_\alpha) = 1$ and $\alpha = J(\psi_\alpha)$ for some 
$\B_i\psi_\alpha \in {\sf sub}(\prop{\varphi})$.
Then, we again close $\propneigh_i$ under intersection so that $\propmodel$ is a $\C^{n}$ model.
%We claim that 
\begin{claim}
If 
$\beta \in \propneigh_i(w)$ and $\beta = J(\bigwedge^{k-1}_{l=1}\psi_l)$ for some $\B_i \psi_1, \ldots, \B_i \psi_{k-1} \in{\sf sub}(\prop{\varphi})$ 
then  $\valuation(\B_i\psi_l)=1$ for all $1\leq l < k$.
\end{claim}
Indeed, since $\beta = J(\bigwedge^{k-1}_{l=1}\psi_l)\in \propneigh_i(w)$ 
we must have that $\nu(\B_i\psi_{1,\beta})=1, \ldots, \nu(\B_i\psi_{m,\beta})=1$ and 
$\beta = \bigcap^{m}_{i=1} J(\psi_{i,\beta})$ for some 
$\B_i\psi_{1,\beta}, \ldots, \B_i\psi_{m,\beta} \in{\sf sub}(\prop{\varphi})$. %\nb{intersection}
Suppose now that $\nu(\bigwedge^{k-1}_{l=1}\psi_l) = 0$. Then, by assumption, there exists a structure
$\propmodel_j = (\propdomain_j, \{ \propneigh_{j_i} \}_{i \in I}, \propassign_j)$ and a world $w_j
\in \propdomain_j$ such that $\propmodel_j,w_j \models (\bigwedge^{k-1}_{l=1}\psi_{l,\beta} \wedge \neg(\bigwedge^{k-1}_{l=1}\psi_l))\vee(\neg(\bigwedge^{k-1}_{l=1}\psi_{l,\beta}) \wedge (\bigwedge^{k-1}_{l=1}\psi_l))$. 
It follows that $\propassign_j(\bigwedge^{k-1}_{l=1}\psi_{l,\beta})\neq \propassign_j(\bigwedge^{k-1}_{l=1}\psi_l)$. 
Consequently $J(\bigwedge^{k-1}_{l=1}\psi_{l,\beta})\neq J(\bigwedge^{k-1}_{l=1}\psi_l)$, which is a contradiction.  
%\nb{O :... talk about consistent}

\medskip

We now show by induction on the structure of formulas 
that $\propassign$ and $J$ agree on ${\sf sub}(\prop{\varphi})$. 
This holds by construction for atomic propositions. It is easy to deal 
with propositional connectives, since we know that $J(\neg \psi)=\propdomain\setminus J(\neg \psi)$
and  $J(\psi_1\wedge \psi_2)=J(\psi_1)\cap J(\psi_2)$ 
and similarly for $\propassign$. Assume inductively that $\propassign(\psi) = J(\psi)$.
Suppose first that $u\in J(\B_i\psi)$. Then, either $u=w$ and $\nu(\B_i\psi)=1$
or $u\in \propdomain_j$ and $\propmodel_j,u\models\B_i\psi$. In either case 
we have that $J(\psi)\in \propneigh_i(u)$. Since 
$\propassign(\psi) = J(\psi)$, it follows that $\propmodel,u\models \B_i\psi$, 
that is, $u\in \propassign(\B_i\psi)$. Suppose now that $u\in \propassign(\B_i\psi)$, 
that is, $\propmodel,u\models \B_i\psi$, or, equivalently, $\propassign(\psi)\in \propneigh_i(u)$.
Since $\propassign(\psi) = J(\psi)$ it follows that either $u=w$ and 
$\nu(\B_i\psi)=1$ or $u\in \propdomain_j$ and $\propmodel_j,u\models \B_i\psi$. 
In either case we have that $u\in J(\B_i\psi)$. 

Since $\nu(\prop{\varphi})=1$, we have that $w\in J(\prop{\varphi})$, 
and consequently $w\in \propassign(\prop{\varphi})$. That 
is, $\propmodel,w\models\prop{\varphi}$. 
The fact that $\propmodel$ is a $\C^{n}$ model follows 
from the definition of $\propneigh_i$.
The fact that $\propmodel$ is a $\varphi$-consistent  model follows from 
the fact that $\nu$, used to construct the assignment 
related to $w$, is $\varphi$-consistent and 
the models $\propmodel_1,\ldots,\propmodel_m$, used to define 
the remaining worlds in $\Wmc$, are all $\varphi$-consistent %$\C^{n}$ 
models. 
%Now the lemma follows from Lemma~\ref{lem:general}.
%Take $\bigwedge^{k-1}_{j=1}\psi_j$ as $\psi'$
\end{proof}

\begin{restatable}{lemma}{Lemmaprop}\label{lem:proplemma}
A formula $\prop{\varphi}$ is satisfied in a $\varphi$-consistent $\Nn$ model
iff
%if, and only if,
there is   a $\varphi$-consistent valuation \valuation 
for $\prop{\varphi}$ such that 
\begin{enumerate}
\item if $\B_i\psi$ is in ${\sf sub}(\prop{\varphi})$ and
$\valuation(\B_i\psi)=0$, then $\neg \psi$ 
is satisfied in a $\varphi$-consistent $\Nn$ model; %and 
\item if $\B_i\psi_1$ and $\B_i\psi_2$ are in 
${\sf sub}(\prop{\varphi})$, $\valuation(\B_i\psi_1)=1$, and 
$\valuation(\B_i\psi_2)=0$, then $(\psi_1\wedge\neg\psi_2)\vee(\neg\psi_1\wedge\psi_2)$
is satisfied in a $\varphi$-consistent $\Nn$ model. 
\end{enumerate}
\end{restatable}
\begin{proof}%[Sketch]
We start with proving ($\Rightarrow$). 
Suppose that $\prop{\varphi}$ is satisfied in a world $w$ of a $\varphi$-consistent \Nn model 
 $\propmodel = (\propdomain, \{ \propneigh_{i} \}_{i \in I}, \propassign)$. That is, 
 $\propmodel, w\models \prop{\varphi}$. We define a $\varphi$-consistent valuation for 
 $\prop{\varphi}$
 by setting $\nu(\psi)=1$ if $\propmodel, w\models \psi$ and $\nu(\psi) = 0$
 if  $\propmodel, w\not\models \psi$. 
 It is easy to check that $\nu$ is indeed a 
 $\varphi$-consistent valuation (given that $\propmodel$ is a  
 $\varphi$-consistent \Nn model). 

\textit{Point 1.}
%\paragraph{Point 1}
By definition,  $\B_i {\sf true}$ holds in $\Nn$ models. %\nb{add this def somewhere}
Since there is a valuation $\valuation$ such that 
$\valuation(\B_i\psi)=0$ we have that
$\psi$ cannot be true in all valuations (otherwise
$\psi\equiv {\sf true}$ would hold and  $\valuation$ 
would violate %validity of 
$\B_i {\sf true}$ in $\Nn$ models).
This means that 
%We have that
 $\neg \psi$ is 
$\Nn$  satisfiable.

\textit{Point 2.}
%\paragraph{Point 2}
Assume that 
 $\B_i\psi_1$ and $\B_i\psi_2$ are in ${\sf sub}(\prop{\varphi})$, 
 $\nu(\B_i\psi_1)=1$ and $\nu(\B_i\psi_2)=0$. Then $\propmodel,w\models \B_i\psi_1$
 and $\propmodel,w\not\models \B_i\psi_2$. Thus, by definition, 
 $\propassign(\psi_1)\in \propneigh_i(w)$ and 
  $\propassign(\psi_2)\not\in \propneigh_i(w)$.
%Using the same argument of Lemma 3.1 in~\cite{Var2}, it follows that 
So,
   $\propassign(\psi_1)\neq \propassign(\psi_2)$.
Then, %it is easy to see that 
there 
   is a world $u$ in the symmetrical difference of these sets 
   such that $\propmodel,u\models (\psi_1\wedge\neg\psi_2)\vee(\neg\psi_1\wedge\psi_2)$. 
   
%\medskip

The proof of the converse ($\Leftarrow$)
%for  the second bullet point 
is as follows. %in Lemma~\ref{lem:general}.
%It holds in $\C^{n}$ models because $\B_i {\sf true}$
 Suppose there is a $\varphi$-consistent valuation $\nu$ for $\prop{\varphi}$ such that 
 Point~1
%\textbf{Point 1}
and
Point~2
%\textbf{Point 2}
hold. 
That is, 
\begin{itemize}
\item if $\B_i\psi$ is in ${\sf sub}(\prop{\varphi})$ and
$\valuation(\B_i\psi)=0$ then 
there is a $\varphi$-consistent \Nn model $\propmodel_{\psi}=(\propdomain_{\psi},
\{ \propneigh_{{\psi}_{i}} \}_{i \in I},\propassign_{\psi})$
 and a world 
$w_{\psi}\in \propdomain_{\psi}$ such that 
$\propmodel_{\psi}, w_{\psi}\models \neg\psi$; and 
%$\neg \psi$ is 
%$\Nn$ satisfiable; and
\item if $\B_i\psi_1$ and $\B_i\psi_2$ are in 
${\sf sub}(\prop{\varphi})$, $\valuation(\B_i\psi_1)=1$, and 
$\valuation(\B_i\psi_2)=0$, then 
there is a $\varphi$-consistent \Nn model $\propmodel_{\psi_1,\psi_2}=(\propdomain_{\psi_1,\psi_2},
\{ \propneigh_{{\psi_1,\psi_2}_{i}} \}_{i \in I},\propassign_{\psi_1,\psi_2})$
 and a world 
$w_{\psi_1,\psi_2}\in \propdomain_{\psi_1,\psi_2}$ such that 
$\propmodel_{\psi_1,\psi_2}, w_{\psi_1,\psi_2}\models (\psi_1\wedge\neg\psi_2)\vee(\neg\psi_1\wedge\psi_2)$. 
\end{itemize}
Let $\propmodel_1,\ldots,\propmodel_m$ be an enumeration of \Nn models 
$\propmodel_{\psi}$ and $\propmodel_{\psi_1,\psi_2}$, as above.  That is, we take one model $\propmodel_{\psi}$ and one model $\propmodel_{\psi_1,\psi_2}$ for each such
subformula $j=\psi$ and pair of subformulas $j=\psi_1,\psi_2$ where $\propmodel_j = (\propdomain_j, \{ \propneigh_{j_{i}} \}_{i \in I},\propassign_j)$, 
and let $w_1,\ldots,w_m$ be an enumeration of the worlds $w_{\psi}$ and $w_{\psi_1,\psi_2}$, 
with $w_j\in \propdomain_j$. Assume without loss of generality that $\propdomain_j\cap \propdomain_k=\emptyset$ 
for $j\neq k$. 
In the following, we define a $\varphi$-consistent \Nn model   $\propmodel = (\propdomain,\{ \propneigh_{i} \}_{i \in I}, \propassign)$ 
for $\prop{\varphi}$. 

Intuitively, we construct $\propmodel$ by taking the union of each 
$\propmodel_j$ with the addition of a new world $w$ that 
will satisfy $\prop{\varphi}$. 
We define $\propdomain$ as $\bigcup_{1\leq j\leq n}\propdomain_j\cup \{w\}$, 
where $w$ is fresh.
%The tricky part of the proof is to define the assignment $\propassign$. 
Before defining $\propneigh_{i}$ and $\propassign$, we define the function $J: {\sf sub}(\prop{\varphi})\rightarrow 2^{\Wmc}$
with $J(\psi)=\bigcup_{0\leq j \leq m} \Vmc_j(\psi)$ for all $\psi\in {\sf sub}(\prop{\varphi})$, where %$I_i$ is as above for $1\leq i\leq n$, 
 %and
$\Vmc_0: {\sf sub}(\prop{\varphi})\rightarrow  2^{\{w\}}$ is the function
that assigns $\psi$ to $\{w\}$, if $\nu(\psi)=1$, 
and to $\emptyset$, otherwise ($\Vmc_j$, for $1\leq j\leq m$, is as above).
By construction, we have that $J(\neg \psi)=\propdomain\setminus J(\psi)$
and $J(\psi_1\wedge \psi_2)=J(\psi_1)\cap J(\psi_2)$. 
We define the assignment $\propassign$ as the function 
$\propassign: \NPr(\varphi)\rightarrow 2^{\Wmc}$ satisfying 
 $\propassign(p_\elaxiom)=J(p_\elaxiom)$ for all $p_\elaxiom\in \NPr(\varphi)$. 
 
It remains to define $\propneigh_i$,
for
$i \in I$.
%$1 \leq i \leq n$. 
For $u\in \propdomain_j$ we put  $\alpha \subseteq \Wmc$ in $\propneigh_i(u)$ 
precisely when $\alpha=\Wmc$, or, $\alpha = J(\psi_{\alpha})$ and
$\propmodel_j,u \models \B_i \psi_{\alpha}$,  
for some $\B_i\psi_{\alpha} \in {\sf sub}(\prop{\varphi})$.
The next two claims establish that $\propneigh_i$ is as expected. 
%We claim that 
\begin{claim}
If $\beta \in \Nmc_i(u)$ and $\beta = J(\psi)$ for some $\B_i\psi\in{\sf sub}(\prop{\varphi})$,
then $\propmodel_j,u\models\B_i\psi$.
\end{claim}
\begin{proof}[Proof of Claim]
Indeed, since $\beta = J(\psi)\in \Nmc_i(u)$, 
we must have that either $\beta=\Wmc$ or $\propmodel_j,u\models \B_i\psi_{\beta}$ and  $\beta = J(\psi_{\beta})$ for 
some $\B_i\psi_\beta \in{\sf sub}(\prop{\varphi})$.
In the former case, as all $\propmodel_j$ models are \Nn models, 
we have that $\propmodel_j,u\models \B_i\psi$.
% (in this case, $\psi$
%is true in all worlds in all \Nn models).
In the latter,  since
$J(\psi)=J(\psi_\beta)$, we also have $\propassign_j(\psi)=\propassign_j(\psi_\beta)$ 
(recall that $\propdomain_j \cap \propdomain_k = \emptyset$ for  
$k \neq j$), so $\propmodel_j,u\models \B_i\psi$ iff 
$\propmodel_j,u \models \B_i \psi_\beta$.
It follows that $\propmodel_j,u\models \B_i\psi$.
\end{proof}

Also, we put $\alpha \subseteq \Wmc$ in $\propneigh_i(w)$ (recall $w$ is the fresh 
world introduced above in $\propdomain$) precisely 
when $\alpha = \Wmc$ or $\nu(\B_i\psi_\alpha) = 1$ and $\alpha = J(\psi_\alpha)$ for some 
$\B_i\psi_\alpha \in {\sf sub}(\prop{\varphi})$.

\begin{claim}
If 
$\beta \in \propneigh_i(w)$ and $\beta = J(\psi)$ for some $\B_i \psi \in{\sf sub}(\prop{\varphi})$ 
then $\nu(\B_i \psi) = 1$.
\end{claim}
\begin{proof}[Proof of Claim]
Indeed, since $\beta = J(\psi)\in \propneigh_i(w)$ 
we must have that either $\beta=\Wmc$
or $\nu(\B_i\psi_\beta)=1$ and $\beta = J(\psi_\beta)$ for some 
$\B_i\psi_\beta \in{\sf sub}(\prop{\varphi})$. 
Suppose  that $\nu(\B_i\psi) = 0$ and $\beta\neq\Wmc$.
Then, by assumption, there exists a $\varphi$-consistent \Nn model
$\propmodel_j = (\propdomain_j, \{ \propneigh_{j_i} \}_{i \in I}, \propassign_j)$ and a world $w_j
\in \propdomain_j$ such that $\propmodel_j,w_j \models (\psi_\beta \wedge \neg\psi)\vee(\neg\psi_\beta \wedge \psi)$. 
It follows that $\propassign_j(\psi_\beta)\neq \propassign_j(\psi)$. 
Consequently $J(\psi_\beta)\neq J(\psi)$, which is a contradiction.  
Now, suppose   that $\nu(\B_i\psi) = 0$ and $\beta=\Wmc$.
By assumption,
%$\neg \psi$ is  satisfiable in a $\varphi$-consistent \Nn model.
%This means that 
there exists a $\varphi$-consistent \Nn model
$\propmodel_j = (\propdomain_j, \{ \propneigh_{j_i} \}_{i \in I}, \propassign_j)$ and a world $w_j
\in \propdomain_j$ such that $\propmodel_j,w_j \models \neg\psi$.
It follows that $\propassign_j(\psi)\neq \Wmc$. Consequently $\beta= J(\psi)\neq \Wmc$, which is a contradiction.  Then,  $\nu(\B_i \psi) = 1$, as required.
%\nb{changing}
%, which means that
%$\beta=\Wmc$.
%\nb{O :... talk about consistent}
\end{proof}

We now show by induction on the structure of formulas 
that $\propassign$ and $J$ agree on ${\sf sub}(\prop{\varphi})$. 
This holds by construction for atomic propositions. It is easy to deal 
with propositional connectives, since we know that $J(\neg \psi)=\propdomain\setminus J(\neg \psi)$
and  $J(\psi_1\wedge \psi_2)=J(\psi_1)\cap J(\psi_2)$ 
and similarly for $\propassign$. Assume inductively that $\propassign(\psi) = J(\psi)$.
Suppose first that $u\in J(\B_i\psi)$. Then, either $u=w$ and $\nu(\B_i\psi)=1$
or $u\in \propdomain_j$ and $\propmodel_j,u\models\B_i\psi$. In either case 
we have that $J(\psi)\in \propneigh_i(u)$. Since 
$\propassign(\psi) = J(\psi)$, it follows that $\propmodel,u\models \B_i\psi$, 
that is, $u\in \propassign(\B_i\psi)$. Suppose now that $u\in \propassign(\B_i\psi)$, 
that is, $\propmodel,u\models \B_i\psi$, or, equivalently, $\propassign(\psi)\in \propneigh_i(u)$.
Since $\propassign(\psi) = J(\psi)$ it follows that either $u=w$ and 
$\nu(\B_i\psi)=1$ or $u\in \propdomain_j$ and $\propmodel_j,u\models \B_i\psi$. 
In either case we have that $u\in J(\B_i\psi)$. 

Since $\nu(\prop{\varphi})=1$, we have that $w\in J(\prop{\varphi})$, 
and consequently $w\in \propassign(\prop{\varphi})$. That 
is, $\propmodel,w\models\prop{\varphi}$. 
The fact that $\propmodel$ is $\varphi$-consistent follows from 
the fact that $\nu$, used to construct the assignment 
related to $w$, is $\varphi$-consistent and 
the models $\propmodel_1,\ldots,\propmodel_m$, used to define 
the remaining worlds in $\Wmc$, are all $\varphi$-consistent.
The fact that  $\propmodel$ contains the unit is by construction, that is,
we defined $\propmodel$
%. That is, 
so that for all 
$i\in [1,n]$ and all $w\in\Wmc$, we have that $\Wmc\in\Nmc_i(w)$. 
Thus, $\propmodel$ is a $\varphi$-consistent \Nn model that satisfies 
$\prop{\varphi}$, as required. 
\end{proof}

To determine satisfiability of $\prop{\varphi}$ in a $\varphi$-consistent model, we use Lemma~\ref{lem:prop} and the characterizations above. 
To establish complexity results, %for the upper bound of $\CALCg$, 
we use the fact that there are only quadratically many   subformulas in $\prop{\varphi}$. 
Satisfiability in
\ALC is \ExpTime-complete  and so, one can determine in exponential time
whether a valuation is $\varphi$-consistent. For an \ExpTime~upper bound, one can
deterministically compute all possible $\varphi$-consistent valuations for 
$(\bigwedge^{k-1}_{j=1}\psi_j\wedge\neg\psi_k)$ (or $(\psi_1 \wedge \neg \psi_2)$) and
decide satisfiability of $\prop{\varphi}$ by a $\varphi$-consistent model using a bottom-up strategy (as in~\cite{Baader:2012:LOD:2287718.2287721}). Since satisfiability in \ALC is \ExpTime-hard, our upper bound is tight.

\begin{theorem}
The $\CALCg$ and $\NALCg$ formula satisfiability problems on constant domain neighbourhood models are \ExpTime-complete.
%Satisfiability in  $\CALCg$  and $\NALCg$ is \ExpTime-complete.
\end{theorem}

%%%%%%%%%%%%%%%%%%%%%%%%%%%%%%%%%%%%%%%%%%%%%%%%%%%%%%%%%%%%%%%%%%%%%%

%\section{Applications and Extensions}
%
%\subsection{Somebody Knows}

%% file: sec_conc.tex
\section{Discussion and Future Work}

%\nb{M: added, to be checked}
In this paper, we have presented first results on reasoning in non-normal modal description logics.
After providing motivations and preliminaries for these logics, we have focused on the following two aspects.
First, we have introduced terminating, sound and complete tableaux algorithms for checking satisfiability of multi-modal description logics formulas in varying domain neighbourhood models based on classes of frames that characterise different non-normal systems, that is, $\mathbf{E}^{n}$, $\mathbf{M}^{n}$, $\mathbf{C}^{n}$, and $\mathbf{N}^{n}$.
We have then studied the complexity of the satisfiability problem restricted to fragments where modal operators can be applied to formulas only (thus without modalised concepts) and interpreted on neighbourhood models with constant domains.
As future work, we plan to investigate along the following directions.
%\nb{M: working here}

%
First, we are interested in adapting our tableau algorithms to check satisfiability of formulas on neighbourhood models with constant domains.
This
%Such an adaptation requires
requires to address the
%problem of the
introduction of fresh variables
%introduced in a certain labelled constraint system
that do not occur in other previously expanded labelled constraints systems.
For instance, by applying the $\mathbf{M}^{n}_{\ALC}$-rules to the $n$-labelled constraint system $S_{n} = \{ n : \Diamond_{i} \exists r. A(x), \Box_{i} \lnot A(x) \}$, we obtain the $m$-labelled constraint system $S_{m} = \{ m : \exists r. A(x),  m : \lnot A(x), m : r(x,y), m : A(y) \}$.
The
%introduction of the
fresh
variable $y$ in $S_{m}$ does not allow us to directly extract a model with constant domain,
%from a completion set with labelled constraint systems of this kind,
since there would be no object in the domain of the world associated with $S_{n}$ capable of representing $y$ correctly.
%the variable $y$ introduced in $S_{m}$.

A possible solution could be to define a suitable notion of \emph{quasimodel}~\cite{GabEtAl03}, to equivalently characterise satisfiability on constant domain neighbourhood models in terms of structures representing ``abstractions'' of the actual models of a formula.
The representation of domain objects across worlds would be given in terms of suitably defined functions, called \emph{runs}, to guarantee that they are well-behaved with respect to their modal properties, and that they do not violate the constant domain assumption.
%These notions could then be used in the soundness proof of the tableau algorithms, where one starts from from a complete and clash-free completion set to construct a quasimodel for a formula (in place of a concrete model), in turn implying its satisfiability.
A similar approach is followed by \citeauthor{SeyErd09}~\cite{SeyErd09} and~\citeauthor{SeyJam09}~\cite{SeyJam09,SeyJam10}, with suitable ``copies'' of worlds introduced to address the problem of the definition of runs.
%representing the behaviour of domain objects across worlds.
In these works, however, it is not made explicit how such a definition should be carried out in detail.
We conjecture that an approach based on \emph{marked variables}, as illustrated in~\citeauthor{GabEtAl03}~\cite{GabEtAl03}, can be fruitfully adopted together with quasimodels to solve the issue of a constant domain model extraction from a complete and clash-free completion set for a formula.

In addition, we are interested in tight complexity results for $\LnALC$ formula satisfiability, with respect to varying and constant domain neighbourhood models.
%This problem requires in particular to develop proof strategies for the lower bounds.
It is known that $\ALC$ formula satisfiability is $\ExpTime$-complete. However, we do not know whether the upper bound for $\LnALC$ formula satisfiability problem on varying or constant domain neighbourhood models can be improved to $\ExpTime$-membership, for any $\mathbf{L} \in \{ \mathbf{E}, \mathbf{M}, \mathbf{C}, \mathbf{N} \}$.
It has to be noted that, at the propositional level, the formula satisfiability problem for the systems $\mathbf{E}$, $\mathbf{M}$, and $\mathbf{N}$ is known to be $\NP$-complete, with a rise to $\PSpace$-completeness for systems containing $\mathbf{C}$~\cite{Var2}.

Finally, we plan to consider satisfiability in other combinations and extensions of non-normal modal description logics. This would naturally lead us to consider both the straightforward cases of $\mathbf{MC}$, $\mathbf{MN}$ and $\mathbf{CN}$ of the classical cube~\cite{LelPim19}, as well as other logics tailored to applications in knowledge representation contexts. In particular, we intend to investigate non-normal modal description logics in epistemic, coalitional, and deontic 
%and conditional 
settings.